\documentclass[a4paper,UKenglish,cleveref, autoref, thm-restate]{lipics-v2021}
\hideLIPIcs  %uncomment to remove references to LIPIcs series (logo, DOI, ...), e.g. when preparing a pre-final version to be uploaded to arXiv or another public repository

\graphicspath{{./figures/}}%helpful if your graphic files are in another directory

\newcommand{\MC}{{\bigcap}}
\newcommand{\tree}{{\mathcal T}}

\title{Capturing the Shape of a Point Set With a Line Segment} %TODO Please add

\author{Nathan van Beusekom}{Department of Mathematics and Computer Science, TU Eindhoven, the Netherlands}{n.a.c.v.beusekom@tue.nl}{https://orcid.org/0000-0003-1813-5299}{}%TODO mandatory, please use full name; only 1 author per \author macro; first two parameters are mandatory, other parameters can be empty. Please provide at least the name of the affiliation and the country. The full address is optional. Use additional curly braces to indicate the correct name splitting when the last name consists of multiple name parts.

\author{Marc van Kreveld}{Department of Information and Computing Sciences, Utrecht University, the Netherlands}{m.j.vankreveld@uu.nl}{https://orcid.org/0000-0001-8208-3468}{}

\author{Max van Mulken}{Department of Mathematics and Computer Science, TU Eindhoven, the Netherlands}{m.j.m.v.mulken@tue.nl}{https://orcid.org/0000-0001-6609-2057}{}

\author{Marcel Roeloffzen}{Department of Mathematics and Computer Science, TU Eindhoven, the Netherlands}{m.j.m.roeloffzen@tue.nl}{https://orcid.org/0000-0002-1129-461X}{}

\author{Bettina Speckmann}{Department of Mathematics and Computer Science, TU Eindhoven, the Netherlands}{b.speckmann@tue.nl}{https://orcid.org/0000-0002-8514-7858}{}

\author{Jules Wulms}{Department of Mathematics and Computer Science, TU Eindhoven, the Netherlands}{j.j.h.m.wulms@tue.nl}{https://orcid.org/0000-0002-9314-8260}{}

\authorrunning{N. van Beusekom et al.} %TODO mandatory. First: Use abbreviated first/middle names. Second (only in severe cases): Use first author plus 'et al.'

\Copyright{Nathan van Beusekom, Marc van Kreveld, Max van Mulken, Marcel Roeloffzen, Bettina Speckmann, and Jules Wulms} %TODO mandatory, please use full first names. LIPIcs license is "CC-BY";  http://creativecommons.org/licenses/by/3.0/

 %TODO mandatory: Please choose ACM 2012 classifications from https://dl.acm.org/ccs/ccs_flat.cfm 

 %TODO mandatory; please add comma-separated list of keywords

\category{} %optional, e.g. invited paper

\relatedversion{} %optional, e.g. full version hosted on arXiv, HAL, or other respository/website
\nolinenumbers %uncomment to disable line numbering

\EventEditors{John Q. Open and Joan R. Access}
\EventNoEds{2}
\EventLongTitle{42nd Conference on Very Important Topics (CVIT 2016)}
\EventShortTitle{CVIT 2016}
\EventAcronym{CVIT}
\EventYear{2016}
\EventDate{December 24--27, 2016}
\EventLocation{Little Whinging, United Kingdom}
\EventLogo{}
\SeriesVolume{42}
\ArticleNo{23}
\newcommand{\tangent}[0]{\ensuremath{\tau}}
\newcommand{\convexchain}[0]{\ensuremath{S}}
\newcommand{\convexhull}[1]{\ensuremath{\mbox{\sf CH}(#1)}}
\newcommand{\OPT}[0]{\ensuremath{\mbox{\sf OPT}}}
\newcommand{\circarc}[0]{\ensuremath{\overset{\frown}{c}}}

\newcommand{\eps}{\varepsilon}

\begin{document}

\maketitle

\begin{abstract}
Detecting location-correlated groups in point sets is an important task in a wide variety of applications areas. In addition to merely detecting such groups, the group's shape carries meaning as well. In this paper, we represent a group's shape using a simple geometric object, a line segment. Specifically, given a radius $r$, we say a line segment is \emph{representative} of a point set $P$ if it is within distance $r$ of each point $p \in P$. We aim to find the shortest such line segment. This problem is equivalent to stabbing a set of circles of radius $r$ using the shortest line segment. We describe an algorithm to find the shortest representative segment in $O(n \log h + h \log^3 h)$ time. Additionally, we show how to maintain a stable approximation of the shortest representative segment when the points in $P$ move.
\end{abstract}

\section{Introduction}

Studying location-correlated groups or clusters in point sets is of interest in a wide range of research areas. There are many algorithms and approaches to find such groups; examples include the well-known \emph{k-means clustering}~\cite{hartigan1979algorithm} or \emph{DBSCAN}~\cite{ester1996density}. In addition to the mere existence of such groups, the group's characteristics can carry important information as well. In wildlife ecology, for example, the perceived shape of herds of prey animals contains information about the behavioral state of animals within the herd~\cite{10.7554/eLife.19505}. Since shape is an abstract concept that can get arbitrarily complex, it is often useful to have a simplified representation of group shape that can efficiently be computed. 
The simplest shape (besides a point) that may represent a group is a line segment, suggesting that the group is stretched in a single direction.

When the points move in the plane, as is the case for animals, the representing line segment may change orientation and length. Also, it may disappear if the shape of the points is no longer captured well by a line segment. Conversely, it can also appear when the points form a segment-like shape again.

Let us concentrate on the static version of the problem first. There are a few simple ways to define a line segment for a set of points. We can use the width of the point set to define a narrowest strip, put tight semi-circular caps on it, and use the centers of these semi-circles as the endpoints of the line segment. We can also use the focal points of the smallest enclosing ellipse, and use them as the endpoints. We can also use a maximum allowed distance from the points to the line segment, and use the shortest line segment possible. The first and third option are based on a hippodrome shape (the Minkowski sum of a line segment and a disk). We note that the second and third option still need a threshold distance to rule out that points are arbitrarily far from the defining line segment. In particular, the case that a line segment is \emph{not} a suitable representation should exist in the model, and also the case where the line segment can become a single point. 
There are multiple other options besides the three given, for example, by using the first eigenvector of the points, or the diameter.

In this paper we study the model given by the third option: given a set $P$ of $n$ points in general position, we want to find the shortest line segment $q_1q_2$ such that all points are within distance $r$.
This model has several advantages: (i) It is a simple model. (ii) It naturally includes the case that no line segment represents the points, or a single point already represents the points. (iii) It guarantees that all points are close to the approximating line segment. (iv) It has desirable properties when the points move: in the first two options, there are cases where the points intuitively remain equally stretched in the direction of the line segment, but points moving orthogonally away from it yields a shorter(!) line segment. This issue does not occur in the chosen model.
Moreover, it was studied before in computational geometry, and we can build on existing algorithmic methods and properties.

The first algorithm published that solves the optimization version of the static problem (in fact, the first option) uses $O(n^4\log n)$ time, for a set of $n$ points~\cite{imai1992}. This was improved by Agarwal et al.~\cite{agarwal1993computing} to $O(n^2\alpha(n)\log^3 n)$, where $\alpha(n)$ is the extremely slowly growing functional inverse of Ackermann's function. The first subquadratic bound was given by Efrat and Sharir~\cite{DBLP:journals/dcg/EfratS96}, who presented an $O(n^{1+\eps})$ time algorithm, for any constant $\eps>0$. They use the fixed-radius version as a subroutine and then apply parametric search. Their fixed-radius algorithm already has the bound of $O(n^{1+\eps})$, as it uses vertical decompositions of a parameter space in combination with epsilon-nets. They remark that their methods can solve the shortest stabber problem for unit disks within the same time, which is our problem.

In this paper we present an improved static result and new kinetic results. We solve the static version in $O(n\log^3 n)$ time by exploiting the geometry of the situation better, which allows us to avoid the use of parametric search and epsilon-nets. Our new algorithm uses a rotating calipers approach where we predict and handle events using relatively simple data structures. We still use a key combinatorial result from~\cite{DBLP:journals/dcg/EfratS96} in our efficiency analysis. For the kinetic problem, we are interested in developing a strategy to maintain a line segment that can not flicker rapidly between ``on'' and ``off'', and whose endpoints move with bounded speed.
To accomplish this, we must relax (approximate) the radius around the line segment in which points can be.
We show that with constant speeds and a constant factor approximation in radius, the endpoints of the line segment move with at a speed bounded by a linear function in $r$, while also avoiding flickering.
These results complement recent results on stability.

\begin{figure}
    \centering
    \includegraphics[page=1]{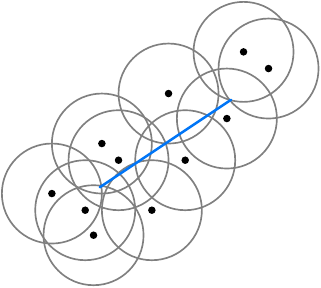}
    \caption{The line segment (blue) must hit every circle of radius $r$, centered at the points in $P$.}
    \label{fig:intuition}
\end{figure}

\subparagraph{Related work.}
A number of shape descriptors have been proposed over the years. A few popular ones are the \emph{alpha shape} of a point set~\cite{edelsbrunner1983shape} and the \emph{characteristic shape}~\cite{duckham2008efficient}, both of which generate representative polygons. Another way to generate the shape of a point set is to fit a function to the point set~\cite{chen2013approximating,hakimi1991fitting, wang2002new}. Bounding boxes and strips are much closer related to the line segment we propose. In the orientation of the first eigenvector, bounding boxes (and strips) have been shown to not capture the dimensions of a point set well~\cite{DBLP:journals/comgeo/DimitrovKKR09}. Optimal bounding boxes and strips, of minimum area and width, respectively, align with a convex hull edge and can be computed in linear time given the convex hull~\cite{DBLP:journals/cacm/FreemanS75,toussaint1983solving}.
Problems of finding one or more geometric objects that intersect a different set of geometric objects are known as \emph{stabbing} problems~\cite{edelsbrunner1982stabbing}, and several variants have been studied~\cite{chan2018stabbing,diaz2015new,schlipf2014stabbing}. As mentioned, stabbing a set of unit circles with the shortest line segment was studied in~\cite{DBLP:journals/dcg/EfratS96}. The inverse variant, line segments stabbed by one or more circles, has also been studied~\cite{claverol2018stabbing,kobylkin2017stabbing}.

Recently, considerable attention has been given to stability of structures under the movement of a set of points or motion of other objects. Stability is a natural concern in, for example, (geo)visualization and automated cartography: In air traffic control planes may be visualized as labeled points on a map, and the labels are expected to smoothly follow the locations of the moving points~\cite{DBLP:conf/esa/BergG13}. Similarly, for interactive maps that allow, for example zooming and panning, labels should not flicker in and out of view~\cite{been2010optimizing,gemsa2011sliding,gemsa2016consistent,nollenburg2010dynamic}. In computational geometry, only the stability of $k$-center problems was studied~\cite{DBLP:conf/dialm/BespamyatnikhBKS00,de2013kinetic,durocher2006steiner,durocher2008bounded}, until Meulemans et al.\ introduced a framework for stability analysis~\cite{meulemans2018framework}. Applying the framework to shape descriptors, they proved that an $O(1)$-approximation of an optimal oriented bounding box or strip moves only a constant-factor faster than the input points~\cite{meulemans2019stability}.

\section{Computing the Shortest Representative Segment}\label{sec:algorithm}

Given a set $P$ of $n$ points and a distance bound $r$, we show how to construct the shortest segment $q_1q_2$ with maximum distance $r$ to $P$. Our algorithm uses the rotating calipers approach~\cite{toussaint1983solving}. We start by finding the shortest representative segment for fixed orientation $\alpha$, after which we rotate by $\pi$ while maintaining the line segment, and return the shortest one we encounter. Note that, even though a representative segment does not exist for every orientation, we can easily find an initial orientation $\alpha$ for which it does exist using rotating calipers; these are the orientations at which the rotating calipers have width $\leq 2r$. Although our input point set $P$ can be of any shape, the following lemma shows that it suffices to consider only its convex hull~$\convexhull{P}$.

\begin{restatable}{lemma}{convexhullLem}
\label{lem:convex-hull}
    If a line segment $q_1q_2$ intersects all circles defined by the points in the convex hull $\convexhull{P}$, then $q_1q_2$ also intersects all circles defined by the points in $P$.
\end{restatable}
\begin{proof}
    Since $q_1q_2$ crosses each circle defined by $\convexhull{P}$, each point in $\convexhull{P}$ has a distance of at most $r$ to $q_1q_2$. The edges of the convex hull are also at most~$r$ to $q_1q_2$, since these are straight line segments between the convex hull points. All other points in $P$ are inside the convex hull and thus each point in $P$ must have a distance of at most $r$ to $q_1q_2$. 
\end{proof}

We can compute $\convexhull{P}$ in $O(n\log h)$ time, where $h$ is the size of the convex hull~\cite{chan1996convexhull}.
% Observe that, if rotating calipers initially finds no orientation with width~$>2r$, then point set $P$ can be enclosed by a circle of radius at most $r$, and the shortest representative segment is the center point of this enclosing circle. \marcel{This is not true right? With an equilateral triangle of side length 2r.}
% Hence, in the rest of this paper we assume that the shortest representative segment has non-zero length. \marcel{Do we use this assumption?}

\subsection{Fixed orientation}
\label{sec:fixed-orientation}
We describe how to find the shortest representative segment with fixed orientation $\alpha$. Using rotating calipers~\cite{toussaint1983solving}, we can find all orientations in which a representative segment exists, and pick $\alpha$ such that a solution exists. For ease of exposition and without loss of generality, we assume $\alpha$ to be horizontal. Let the \emph{left} and \emph{right half-circle} of a circle $C$ be the half-circle between $\pi/2$ and $3\pi/2$ and between $3\pi/2$ and $5\pi/2$, respectively. Lemma~\ref{lem:convex-hull} permits us to consider only points of~$P$ on the convex hull, thus for the remainder of this paper we use $\mathcal{C}_P$ to indicate the set of circles of radius $r$ centered at the points of $P$ in $\convexhull{P}$. 

Observe that every horizontal line that lies below the bottom-most top horizontal tangent $\tangent_1$ and above the top-most bottom horizontal tangent $\tangent_2$ of all circles crosses all circles (see Figure~\ref{fig:strip}). If $\tangent_1$ lies below $\tangent_2$, then there exists no horizontal line that crosses all circles. 

\begin{figure}
    \centering
    \includegraphics[page=4]{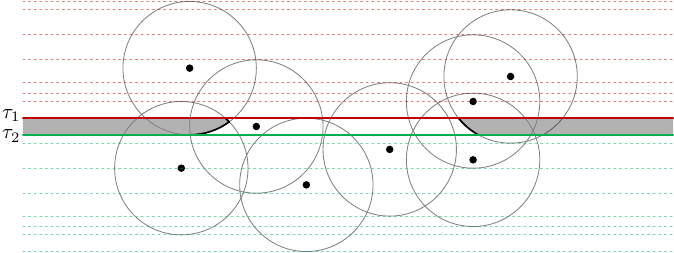}
    \caption{Two extremal tangents~$\tangent_1$ and $\tangent_2$ for horizontal orientation~$\alpha$. The shortest line segment of orientation~$\alpha$ that intersects all circles, ends at the boundary of the gray regions. }
    \label{fig:strip}
\end{figure}

To place $q_1q_2$ in the strip between $\tangent_1$ and $\tangent_2$, we can define regions $R_1, R_2$ in which endpoints $q_1$ and $q_2$ must be placed such that $q_1q_2$ intersects all circles (see Figure~\ref{fig:strip}). 

The region $R_1$ is defined as the set of points below or on $\tangent_1$ and above or on $\tangent_2$ and right or on the right-most envelope of all left half-circles. The region $R_2$ is defined analogously using the left envelope of right half-circles. We use $\convexchain_1$ and $\convexchain_2$ to denote the envelope boundary of $R_1$ and $R_2$ respectively. Note that $\convexchain_1$ and $\convexchain_2$ are convex and consist of circular arcs from the left and right half-circles respectively. 
%The boundaries of $R_1$ and $R_2$ are defined by a \emph{convex sequence} $\convexchain_1$ and $\convexchain_2$ of the right-most left circles and left-most right circles, delimited by the two tangents $\tangent_1, \tangent_2$, as well as these tangents themselves. 
If $R_1$ and $R_2$ intersect, then we can place a single point in their intersection at distance at most $r$ from all points in $P$. 
Otherwise, note that $q_1$ and $q_2$ must be on the convex sequences $\convexchain_1$ and $\convexchain_2$, respectively; otherwise, we can move the endpoint onto the convex sequence, shortening $q_1q_2$ and still intersecting all~circles.

We will show that we can compute $\convexchain_1$ and $\convexchain_2$ in $O(h)$ time. First, we show that the half-circles on a convex sequence appear in order of the convex hull.

\begin{restatable}{lemma}{circleOrder}
\label{lem:circle-order}
    The order of the circular arcs in $\convexchain_1$ or $\convexchain_2$ matches the order of their corresponding centers in $\convexhull{P}$.
\end{restatable}
\begin{proof}    
    Consider an arbitrary circular arc $\circarc_k$ centered in $p_k$ on $\convexchain$. It has at most two intersections on $\convexchain$, with $\circarc_{k-1}$ and $\circarc_{k+1}$. Let $\circarc_{k}$ be traversed clockwise from $\circarc_{k-1}$ to $\circarc_{k+1}$. 
    We show that $p_k$ must lie in between $p_{k-1}$ and $p_{k+1}$ in clockwise order in $\convexhull{P}$.
    Consider the triangle $(p_{k-1},p_{k},p_{k+1})$, and consider, invariant to instance rotation, the horizontal edge $p_{k+1}p_{k-1}$ with $p_{k+1}$ on the left and $p_{k-1}$ on the right (see Figure~\ref{fig:sequence-order-triangle}).
    Since $\circarc_{k}$ lies between $\circarc_{k-1}$ and $\circarc_{k+1}$ on $\convexchain$, $p_k$ must lie below the edge $p_{k+1}p_{k-1}$, otherwise $\circarc_{k}$ could not be on $\convexchain$.
    Hence, the points on the triangle $(p_{k-1},p_{k},p_{k+1})$ occur that way in clockwise order. 
    Since the point $p_k$ lies in between $p_{k-1}$ and $p_{k+1}$ in clockwise order in $(p_{k-1},p_{k},p_{k+1})$, it must also be in that order in the convex hull, otherwise the convex hull would have a concave angle. By transitivity, the order of circular arcs in~$\convexchain$ matches the order of the corresponding centers in~$\convexhull{P}$.\qedhere

    \begin{figure}
    \centering
    \includegraphics{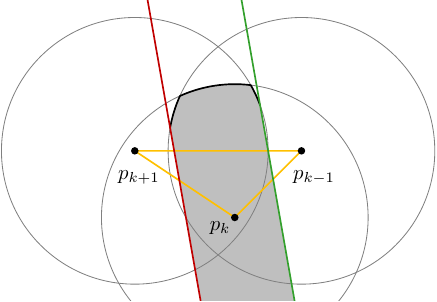}
    \caption{The triangle $(p_{k-1},p_{k},p_{k+1})$ (yellow) such that the edge $p_{k+1}p_{k-1}$ is horizontal. It must hold that $p_k$ lies below $p_{k+1}p_{k-1}$.}
    \label{fig:sequence-order-triangle}
\end{figure}    
\end{proof}

Now we can compute the convex sequences in linear time, given the tangents~$\tangent_1$ and~$\tangent_2$, which can easily be found in linear time.

\begin{restatable}{lemma}{sequenceConstruction}
\label{lem:sequence-construction}
    Given tangents $\tau_1$ and $\tau_2$, and $\convexhull{P}$, we can construct $\convexchain_1$ and $\convexchain_2$ in~$O(h)$.    
\end{restatable}
\begin{proof}
    We may assume that a solution exists. This can easily be checked in $O(h)$ time.
    We describe only the construction of $\convexchain_1$, as $\convexchain_2$ can be constructed symmetrically. Without loss of generality, assume that $\tangent_1$ denotes the start of $\convexchain_1$ in clockwise order. We can find the first arc on $\convexchain_1$ by checking all intersections between $\tangent_1$ and the relevant half-circles, and identifying the most extremal intersection in $O(h)$ time (see Figure~\ref{fig:strip}) We add the part of the circle that lies between $\tangent_1$ and $\tangent_2$ to $\convexchain_1$. We then process each point $p_i$ along the convex hull in clockwise order from the point defining our initial arc.

    Next, let $\circarc_1, \ldots, \circarc_k$ denote the circular arc pieces for $S_1$ constructed so far. Let $p_i$ be the next point on the convex hull that we process. Let $c_i$ denote the left half-circle centered at $p_i$ and let $c_k$ be the support left half-circle of $\circarc_k$. We find the intersection between $c_i$ and $c_k$. If there is no intersection, then $c_i$ must lie entirely to the left of $c_k$ and it cannot contribute to $S_1$. If the intersection point is below $\tangent_2$ then between $\tangent_1$ and $\tangent_2$ we have that $c_i$ lies left of $c_k$ and it cannot contribute to $S_1$. If the intersection point lies within $\circarc_k$ then we update $S_1$ to switch at the intersection point from $c_k$ to $c_i$ as then $c_i$ must lie right of $c_k$ below the intersection point. If the intersection point lies above $\circarc_k$ then the entirety of $\circarc_k$ lies to the left of $c_i$, therefore $\circarc_k$ cannot contribute to $S_1$ and we can discard $\circarc_k$. We then continue by comparing $c_i$ to $c_{k-1}$.
    
    Whenever a half-circle is possibly added it is compared to at most $O(|\convexchain_1|)$ arcs. However, when the half-circle is compared to $i$ arcs, then $i-1$ arcs would be removed from $\convexchain_1$. Thus, by an amortization argument, this happens $O(h)$ times. 
\end{proof}

Next, we must place $q_1$ and $q_2$ on $\convexchain_1$ and $\convexchain_2$, respectively, such that $q_1q_2$ is shortest. We show that $q_1q_2$ is the shortest line segment of orientation $\alpha$ when the tangents of $\convexchain_1$ at $q_1$ and $\convexchain_2$ at $q_2$ are equal. 
Vertices on $\convexchain_1$ and $\convexchain_2$ have a range of tangents (see Figure~\ref{fig:tangent-mapping}).

\begin{figure}
    \centering
    \includegraphics{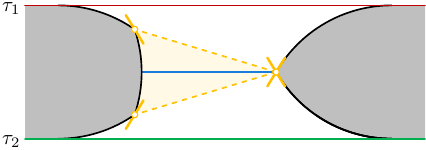}
    \caption{
    Two convex sequences between $\tangent_1$ and $\tangent_2$. There are multiple points on the left convex sequence that have the same tangent as the right yellow vertex. Still, there is only one line segment in horizontal orientation for which the tangents of its endpoints are equal (blue).}
    \label{fig:tangent-mapping}
\end{figure}

\begin{restatable}{lemma}{shortestMapping}
\label{lem:shortest-mapping}
    Let $\convexchain_1$ and $\convexchain_2$ be two convex sequences of circular arcs, and let $q_1$ and $q_2$ be points on $\convexchain_1$ and $\convexchain_2$, respectively, such that line segment $q_1q_2$ has orientation $\alpha$. If the tangent on $\convexchain_1$ at $q_1$ is equal to the tangent on $\convexchain_2$ at $q_2$, then $q_1q_2$ is minimal. 
\end{restatable}
\begin{proof}
    Assume that we are in a situation where the tangents on $\convexchain_1$ at $q_1$ and on $\convexchain_2$ at $q_2$ are not equal, and assume without loss of generality that moving $q_1$ and $q_2$ closer to $\tangent_2$, while keeping $q_1q_2$ in orientation~$\alpha$, lowers the difference in slope between the tangents on $\convexchain_1$ and $\convexchain_2$. 
    This is always possible, because the slopes of the tangents along $\convexchain_1$ and $\convexchain_2$ change in opposite directions from $\tangent_1$ to $\tangent_2$. Next, consider two cases of how the line segment $q_1^*q_2^*$ with equal tangents can look. We show that moving towards $\tangent_2$, and thus towards $q_1^*q_2^*$, makes $q_1q_2$ shorter. First, the slopes of the tangents on $\convexchain_1$ at $q_1^*$ and on $\convexchain_2$ at $q_2^*$ are perpendicular to~$\alpha$. In this case, the line segment $q_1q_2$ shrinks on both sides when moving towards~$\tangent_2$. Second, the slopes of the tangents on $\convexchain_1$ at $q_1^*$ and on $\convexchain_2$ at $q_2^*$ are not perpendicular to~$\alpha$. As a result, by moving towards $\tangent_2$ the line segment $q_1q_2$ shrinks on at least one side. If the $q_1q_2$ shrinks on both sides, it naturally becomes shorter, so assume $q_1q_2$ shrinks on only one side. Compared to the situation where the slopes are equal, by our assumption, at~$q_1$ and~$q_2$ one slope is flatter and one is steeper. If moving~$q_1$ and $q_2$ towards $\tangent_2$ decreases the difference in slopes, the steeper slope must be located on the side of $q_1q_2$ that shrinks, since moving $q_1$ and $q_2$ flattens that slope. Hence, moving $q_1q_2$ towards $\tangent_2$ lowers the length of $q_1q_2$.
\end{proof}
 
Observe that the length of $q_1q_2$ is unimodal between~$\tangent_1$ and $\tangent_2$. We can hence binary search in $O(\log h)$ time for the optimal placement of $q_1$ and $q_2$. By Lemmata~\ref{lem:sequence-construction} and~\ref{lem:shortest-mapping} we can compute the shortest representative segment of fixed orientation $\alpha$ in $O(h)$ time.

\subsection{Rotation}
\label{sec:rotation}
After finding the shortest line segment for a fixed orientation $\alpha$, as described in the previous section, we sweep through all orientations~$\alpha$ while maintaining $\tangent_1$, $\tangent_2$, $\convexchain_1$, $\convexchain_2$, and the shortest representative segment $q_1q_2$ of orientation $\alpha$. We allow all of these maintained structures to change continuously as the orientation changes, and store the shortest representative segment found. Any time a discontinuous change would happen, we trigger an \emph{event} to reflect these changes. We pre-compute and maintain a number of \emph{certificates} in an event queue, which indicate at which orientation the next event occurs. This way we can perform the continuous motion until the first certificate is violated, recompute the maintained structures, repair the event queue, and continue rotation.

We distinguish five types of events:
\begin{enumerate}
    \item $q_1$ or $q_2$ moves onto/off a vertex of $\convexchain_1$ or $\convexchain_2$;
    \item $\tangent_1$ or $\tangent_2$ is a bi-tangent with the next circle on the convex hull;
    \item $\tangent_1$ and $\tangent_2$ are the same line;
    \item $\tangent_1$ or $\tangent_2$ is tangent to $\convexchain_1$ or $\convexchain_2$ and rotates over a (prospective) vertex of $\convexchain_1$ or $\convexchain_2$;
    \item $\tangent_1$ or $\tangent_2$ is not tangent to $\convexchain_1$ or $\convexchain_2$ and rotates over a (prospective) vertex of $\convexchain_1$ or $\convexchain_2$.
\end{enumerate}

Since the shortest line segment~$q_1q_2$ in orientation~$\alpha$ is completely determined by $\tangent_1$, $\tangent_2$, $\convexchain_1$, and $\convexchain_2$, the above list forms a complete description of all possible events. Thus, we maintain at most two certificates for events of type 1 (one for each convex sequence) and 2 (one for each tangent), and a single type-3 certificate. Additionally, there must be exactly one type-4 or type-5 certificate for each endpoint of $\convexchain_1$ and $\convexchain_2$, so four in total. These are stored in a constant-size event queue~$Q$, ordered by appearance orientation. Insert, remove, and search operations on $Q$ can hence be performed in $O(1)$~time. 

We will describe below how all events over a full rotational sweep can be handled in $O(h\log^3 h)$ time in total. Combined with the computation of the convex hull of $P$ this yields the following theorem. Note that in the worst case $P$ is in convex position, and $n = h$.

\begin{restatable}{theorem}{staticFull}
\label{thm:static-full}
    Given a point set $P$ consisting of $n$ points and a radius $r$, we can find the shortest representative segment in $O(n\log h + h\log^3h)$ time, where $|\convexhull{P}| = h$. 
\end{restatable}

\subparagraph{Event handling.}

In the following descriptions, we assume that an event happens at orientation $\alpha$, and that $\varepsilon$ is chosen such that no other events occur between $\alpha - \varepsilon$ and $\alpha + \varepsilon$. We also assume that no two events happen simultaneously, which is a general position assumption. We describe, for each event type, the time complexity of computing a new certificate of that type, the time complexity of resolving the event, and the number of occurrences. 

\subparagraph{(1) $q_1$/$q_2$ moves onto/off of a vertex of $\convexchain_1$/$\convexchain_2$.}

We describe, without loss of generality, how to handle the event involving $q_1$ and $\convexchain_1$; the case for $q_2$ and $\convexchain_2$ is analogous. See Figure~\ref{fig:event-1} for an example of this event. First, observe that we can compute certificates of this type in $O(1)$ time, simply by walking over $\convexchain_1$ to find the next vertex/arc $q_1$ should move onto. 

\begin{observation}\label{obs:event1-cert}
    We can construct a new certificate of type 1 in $O(1)$ time.
\end{observation}

Observe that, since vertices of $\convexchain_1$ cover a range of tangents, there are intervals of orientations at which $q_1$ remains at a vertex of $\convexchain_1$. As such, we describe two different cases for this event: $q_1$ moves \emph{onto} or \emph{off} a vertex of $\convexchain_1$.

If $q_1$ was moving over an arc of $\convexchain_1$ at $\alpha - \varepsilon$ and encounters a vertex at $\alpha$, then the movement path of $q_1$ is updated to remain on the encountered vertex. Additionally, we place a new type-1 certificate into the event queue that is violated when $q_1$ should move off the vertex, when the final orientation covered by the vertex is reached.

If $q_1$ is at a vertex at $\alpha - \varepsilon$ and orientation $\alpha$ is the final orientation covered by that vertex, then the movement path of $q_1$ must be updated to follow the next arc on $\convexchain_1$. Additionally, we place a new type-1 certificate into the event queue that is violated when $q_1$ encounters the next vertex, at the orientation at which this arc of $\convexchain_1$ ends. 

\begin{figure}
    \centering
    \includegraphics[page=5]{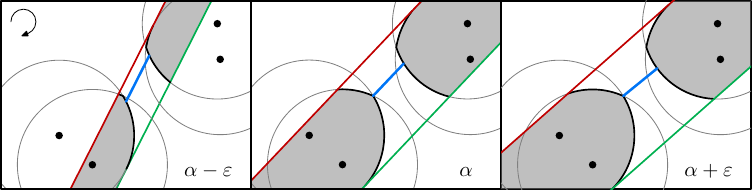}
    \caption{When $q_1$/$q_2$ is at a vertex of $\convexchain_1$/$\convexchain_2$, it stops moving.}
    \label{fig:event-1}
\end{figure}

\begin{restatable}{lemma}{eventOne}
\label{lem:event1}
    Throughout the full $\pi$ rotation, type-1 events happen at most $O(h)$ times, and we can resolve each occurrence of such an event in $O(1)$ time.
\end{restatable}
\begin{proof}
    Since we sweep over orientations in order and $\convexchain_1$ is convex, $q_1$ moves monotonically over $\convexchain_1$ unless $q_1$ lies on $\tangent_1$. In that case, $q_1$ can temporarily move backwards over $\convexchain_1$ until it encounters the previous vertex, and then start moving monotonically forward again. This can only happen once per convex hull vertex. This means that each convex hull vertex $v$ induces at most a constant number of events (two for each vertex adjacent to the arc corresponding to $v$ on $\convexchain_1$, and two for moving backwards over a vertex of $v$ when $\tangent_1$ corresponds to $v$). As such, this event happens at most $O(h)$ times. 

    Assigning a new movement path to $q_1$ or $q_2$ can be done in $O(1)$ time. Constructing a new type-1 certificate takes $O(1)$ time by Observation~\ref{obs:event1-cert}. Inserting a new certificate into the constant-size event queue takes $O(1)$ time.
\end{proof}

\subparagraph{(2) $\tangent_1$ or $\tangent_2$ is bi-tangent with the next circle on the convex hull.}

We describe, without loss of generality, how to handle the event involving $\tangent_1$; handling $\tangent_2$ is analogous. See Figure~\ref{fig:event-2} for an illustration. First, observe that we can compute certificates of this type in $O(1)$ time, since these certificates depend only on the orientation of the next convex hull edge. 

\begin{observation}\label{obs:event2-cert}
    We can construct a new certificate of type 2 in $O(1)$ time.
\end{observation}

When $\tangent_1$ is a bi-tangent of two circles defined by their centers $u,v\in P$ then, by definition of $\tangent_1$, $u$~and $v$ must both be the extremal points in the direction $\theta$ perpendicular to~$\alpha$. Therefore, $(u,v)$ must be an edge on the convex hull. Suppose that, without loss of generality, $u$~was the previous extremal vertex in direction $\theta-\varepsilon$, then $v$~is extremal in direction $\theta+\varepsilon$. As such, $\tangent_1$ belongs to $u$ at $\alpha-\varepsilon$, and to $v$ at $\alpha+\varepsilon$. When this happens, we insert a new type-2 certificate into the event queue that is violated at the orientation of the next convex hull edge. Additionally, we must recompute the certificates of type 3, 4 and 5 that are currently in the event queue, since these are dependent on $\tangent_1$. 

\begin{figure}
    \centering
    \includegraphics[page=1]{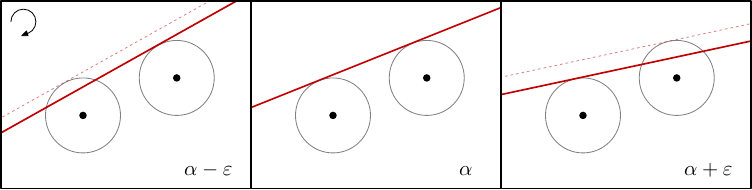}
    \caption{When the defining circle of $\tangent_1$/$\tangent_2$ changes, $\tangent_1$/$\tangent_2$ is parallel to a convex hull edge.}
    \label{fig:event-2}
\end{figure}

\begin{restatable}{lemma}{eventTwo}
\label{lem:event2}
    Throughout the full $\pi$ rotation, type-2 events happen at most $O(h)$ times, and we can resolve each occurrence of such an event in $O(\log^2 h)$ time.
\end{restatable}
\begin{proof}
    Since there are $O(h)$ convex hull edges, and this event happens at most once per convex hull edge, this event occurs at most $O(h)$ times. By Observations~\ref{obs:event2-cert},~\ref{obs:event3-cert} and~\ref{obs:event4-cert} and Lemma~\ref{lem:event5-cert}, we can construct a new type-2, type-3, type-4 and type-5 certificates in $O(1)$ and $O(\log^2 h)$ time. Deletions from and insertions into the event queue take $O(1)$ time. 
\end{proof}

\subparagraph{(3) $\tangent_1$ and $\tangent_2$ are the same line.}

When this event takes place, $\tangent_1$ and $\tangent_2$ are the inner bi-tangents of their two respective defining circles. See Figure~\ref{fig:event-3} for an example. First, observe that we can compute certificates of this type in $O(1)$ time by simply finding the inner bi-tangent of the circles corresponding to $\tangent_1$ and $\tangent_2$. 

\begin{observation}\label{obs:event3-cert}
    We can construct a new certificate of type 3 in $O(1)$ time. 
\end{observation}

We distinguish two different cases for this event: either there is a solution at $\alpha-\varepsilon$ and no solution at $\alpha+\varepsilon$, or vice versa. 

If there was a solution at $\alpha-\varepsilon$ and there is none at $\alpha+\varepsilon$, we simply stop maintaining $q_1q_2$, $\convexchain_1$ and $\convexchain_2$ until there exists a solution again. As such, we remove all type-1, type-5 and type-4 certificates from the event queue and place a new type-3 certificate into the event queue that is violated at the next orientation where $\tangent_1$ and $\tangent_2$ are the same line.

If there was no solution at $\alpha-\varepsilon$ and there is a solution at $\alpha+\varepsilon$, we must recompute $\convexchain_1$, $\convexchain_2$, and $q_1q_2$ at orientation~$\alpha$. At orientation~$\alpha$, $\convexchain_1$ and $\convexchain_2$ are single vertices where $\tangent_1$ and $\tangent_2$ intersect the extremal half-circles of the arrangement. Then, $q_1q_2$ is the line segment between these single vertices of $\convexchain_1$ and $\convexchain_2$. We place new type-1, type-4 and type-5 certificates into the event queue reflecting the newly found $\convexchain_1$, $\convexchain_2$, $q_1$ and $q_2$. Additionally, we insert a new type-3 certificate that is violated at the next orientation where $\tangent_1$ and $\tangent_2$ are the same line.

\begin{figure}
    \centering
    \includegraphics[page=2]{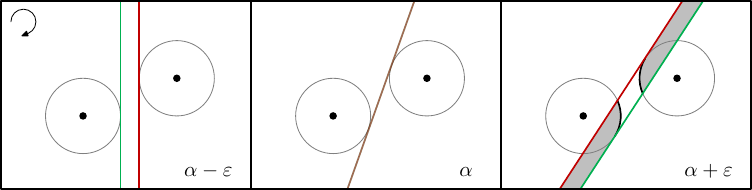}
    \caption{When $\tangent_1$ and $\tangent_2$ are the same line, they are an inner bi-tangent of their two defining~circles.}
    \label{fig:event-3}
\end{figure}

\begin{restatable}{lemma}{eventThree}
\label{lem:event3}
    Throughout the full $\pi$ rotation, type-3 events happen at most $O(h)$ times, and we can resolve each occurrence of such an event in $O(\log^2 h)$ time.
\end{restatable}
\begin{proof}
    Observe that $\tangent_1$ and $\tangent_2$ form an inner bi-tangent of their respective circles when they are the same line (see Figure~\ref{fig:event-3}), and each pair of circles has two such inner bi-tangents. Hence, throughout the the full $\pi$ rotation $\tangent_1$ and $\tangent_2$ can be the same line only $O(h)$ times.

    If there was a solution at $\alpha - \varepsilon$, we compute a new type-3 certificate, which can be done in $O(1)$ by Observation~\ref{obs:event3-cert}. Then, we remove and insert certificates to/from a constant size event queue, which can be done in $O(1)$ time. 

    If there was no solution in $\alpha-\varepsilon$, then, symmetrically, there must be a solution in $\alpha+\varepsilon$. We can find the two arcs that contain $q_1$ and $q_2$ by finding the outermost the intersections between $\tangent_1$ (or $\tangent_2$) and circles in $\mathcal{C}_P$. To do this efficiently, we require an additional data structure, which is described later in Section~\ref{sec:ds}. By Lemma~\ref{lem:query-line}, we can query this data structure in $O(\log^2 h)$ time to obtain starting vertices for $\convexchain_1$ and $\convexchain_2$. 
    
    Next, we recompute the type-1 and type-3 certificates, which by Observations~\ref{obs:event1-cert} and~\ref{obs:event3-cert} can be done in $O(1)$ time. Lastly, we compute a constant number of new type-4 or type-5 certificates, depending on whether the circle corresponding to $\tangent_1$ or $\tangent_2$ is the outer arc of $\convexchain_1$ or $\convexchain_2$. By Observation~\ref{obs:event4-cert} and Lemma~\ref{lem:event5-cert}, this can be done in $O(1)$ and $O(\log^2 h)$ time. Adding these new certificates into the constant size event queue takes $O(1)$ time. 
\end{proof}

\subparagraph{(4) $\tangent_1$ or $\tangent_2$ is tangent to $\convexchain_1$ or $\convexchain_2$ and rotates over a (prospective) vertex of $\convexchain_1$ or $\convexchain_2$.}
We describe, without loss of generality, how to handle the event involving $\tangent_1$ and $\convexchain_1$; the case for $\tangent_2$ and $\convexchain_2$ is analogous. See Figure~\ref{fig:event-4} for an example of this event. First, observe that we can compute certificates of this type in $O(1)$ time: Let $C_i$ be the circle to which $\tangent_1$ is tangent. Then, to construct a certificate, we find the orientation at which $\tangent_1$ hits the intersection point of $C_i$ and $\convexchain_1$. Note that this intersection is part of~$\convexchain_1$ or will appear at~$\tau_1$.

\begin{observation}\label{obs:event4-cert}
    We can construct a new certificate of type 4 in $O(1)$ time.
\end{observation}

Let vertex $v$ be the vertex of the convex chain $\convexchain_1$ that is intersected by $\tangent_1$ at orientation~$\alpha$. Then either vertex $v$ is a vertex of $\convexchain_1$ at orientation $\alpha - \varepsilon$ but not at $\alpha + \varepsilon$, or vice~versa.

In the prior case, at orientation $\alpha$ the arc to which $\tangent_1$ is a tangent is completely removed from $\convexchain_1$. Vertex $v$ becomes the endpoint of $\convexchain_1$ and starts moving along the next arc of $\convexchain_1$. If the affected arc or vertex appeared in a type-1 certificate in the event queue, it is updated to reflect the removal of the arc and the new movement of the vertex. Additionally, we place a new type-5 certificate into the event queue.

In the latter case, at orientation $\alpha$ an arc of the incident circle to $\tangent_1$ needs to be added to $\convexchain_1$. If the arc that was previously the outer arc of $\convexchain_1$ appeared in a type-1 certificate in the event queue, it may need to be updated to reflect the addition of the new arc. Additionally, we place a new type-4 certificate into the event queue.

\begin{figure}
    \centering
    \includegraphics[page=4]{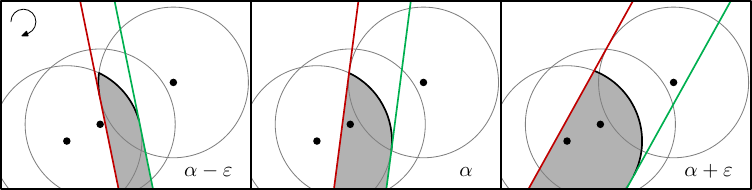}
    \caption{When $\tangent_1$/$\tangent_2$ hits an intersection of its defining circle that is also on $\convexchain_1$/$\convexchain_2$, an arc is removed from $\convexchain_1$/$\convexchain_2$.}
    \label{fig:event-4}
\end{figure}

\begin{restatable}{lemma}{eventFour}
\label{lem:event4}
    Throughout the full $\pi$ rotation, type-4 events happen at most $O(h)$ times, and we can resolve each occurrence of such an event in $O(\log^2 h)$ time.
\end{restatable}
\begin{proof}
    Since the convex hull vertex corresponding to $\tangent_1$ changes in order of the convex hull, over a full rotational sweep $\tangent_1$ can correspond to each convex hull vertex only once. Additionally, the rotational movement direction of the intersection between $S_1$ and $\tangent_1$ reverses only with events of type 4, or when the convex hull vertex corresponding to $\tangent_1$ changes. As such, events of type 4 can only happen once per convex hull vertex, of which there are $O(h)$.

    To handle a type-4 event, we must add/remove the circular arc corresponding to $\tangent_1$ to $S_1$, which we can do in $O(1)$. Additionally, we compute a new type-1 certificate, as well as a new type-4 or type-5 certificate, which can be done in $O(1)$ and $O(\log^2h)$ time by Observations~\ref{obs:event1-cert} and~\ref{obs:event4-cert} and Lemma~\ref{lem:event5-cert}, respectively. Then, we remove the old type-1 certificate from the event queue, and add the new certificates, all in $O(1)$ time.
\end{proof}

\subparagraph{(5) $\tangent_1$ or $\tangent_2$ is not tangent to $\convexchain_1$ or $\convexchain_2$ and rotates over a (prospective) vertex of $\convexchain_1$ or $\convexchain_2$.}
We describe, without loss of generality, how to handle the event involving $\tangent_1$ and $\convexchain_1$; the case for $\tangent_2$ and $\convexchain_2$ is analogous. See Figure~\ref{fig:event-5} for an example of this event. The time complexity of constructing a certificate of this type is stated in the following lemma.

\begin{restatable}{lemma}{eventFivecert}
\label{lem:event5-cert}
    We can construct a new certificate of type 5 in $O(\log^2h)$ time.
\end{restatable}

Additionally, we get the following bounds on handling type-5 events.

\begin{restatable}{lemma}{eventFive}
\label{lem:event5}
    Throughout the full $\pi$ rotation, $\tangent_1$ or $\tangent_2$ hits a vertex of $\convexchain_1$ or $\convexchain_2$ at most $O(h \log h)$ times, and we handle each occurrence of this event in $O(\log^2h)$ time.
\end{restatable}

We prove Lemmata~\ref{lem:event5-cert} and~\ref{lem:event5} in the following section, using an additional data structure which is described in detail in Section~\ref{sec:ds}.

\begin{figure}
    \centering
    \includegraphics[page=3]{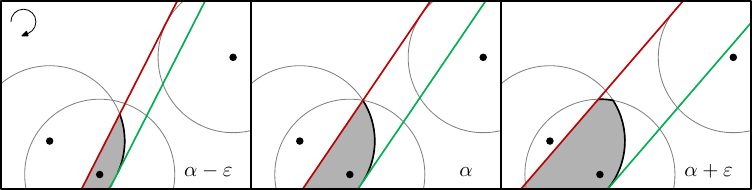}
    \caption{When $\tangent_1$/$\tangent_2$ hits an intersection of two circles, an arc needs to be added to $\convexchain_1$/$\convexchain_2$.}
    \label{fig:event-5}
\end{figure}

\subsection{Finding and maintaining the convex sequence}
\label{sec:event5}

In this section, we describe an additional data structure necessary to maintain the convex sequences $\convexchain_1$ and $\convexchain_2$ efficiently. We can use this data structure to construct and handle violations of type-5 certificates efficiently, as well as to find new starting positions of $\convexchain_1$ and $\convexchain_2$ after a type-3 event.

Let $p_1, \ldots, p_h$ be the vertices of the convex hull in clockwise order. At a given orientation~$\alpha$, let $p_i$ be the point corresponding to the circle $C_i$ to which $\tangent_1$ is tangent. We use $v_\tangent$ to denote the intersection point between $\tangent_1$ and $\convexchain_1$, if it exists, which is simultaneously an endpoint of $S_1$. Let $p_j$ be the point corresponding to the circle $C_j$ on which $v_\tangent$ is located. This implies that the arc on $S_1$ intersected by $\tangent_1$ belongs to circle $C_j$. Then, during our rotational sweep, $v_\tangent$ is moving over $C_j$. A type-5 event takes place when $v_\tangent$ hits the intersection of $C_j$ with another circle $C_k$ corresponding to point $p_k$. 

If, before a type-5 event, the arc of $C_j$ on $S_1$ was shrinking due to the movement of $v_\tangent$, then $C_j$ is fully removed from $S_1$ at the event, and $v_\tangent$ continues moving over $S_1$. Constructing the certificate in this case is very easy, since all we need to do is walk over $S_1$ from $v_\tangent$ to find the next vertex. As such, for the remainder of this section, we consider only the more complicated type-5 event, where the arc of $C_j$ on $S_1$ is growing due to the movement of $v_\tangent$. 

In that case, when the type-5 event happens, an arc of $C_k$ is added to $S_1$, and $v_\tangent$ starts moving over $C_k$ instead of $C_j$. As such, to construct a type-5 certificate, we must find the intersection between $C_j$ and another circle $C_k$ belonging to a point $p_k \in P$, such that the intersection between $C_j$ and $C_k$ is the first intersection hit by $v_\tangent$. To do this, we will first state some characteristics of $C_k$ and $p_k$.

First, observe that simply finding the first intersecting circle $C_k$ of $C_j$ is not necessarily enough. We are only interested in the semi-circles of all circles in $\mathcal{C}_P$ that have the same `opening direction' as the convex chain $\convexchain_1$ for a given orientation $\alpha$. As such, let the \emph{convex semi-circle} of a given circle $C$ be the semi-circle of $C$ that is convex with respect to $S_1$. Conversely, let the \emph{concave semi-circle} of $C$ be the opposite semi-circle of $C$. Then, circle $C_k$ appears on $S_1$ after a hit by $v_\tangent$ at orientation $\alpha$ if $v_\tangent$ is on the convex semi-circle of $C_k$ at orientation $\alpha$. If this is not the case, $C_k$ should be skipped. We show that, for this reason, we never have to consider points $p_l$, with $l < i$ or $j < l$, when constructing a type-5 certificate. 

\begin{restatable}{lemma}{betweenij}
    \label{lem:between_ij}
    Let $C_i$ be the circle defining $\tangent_1$, and let $C_j$ be the circle on which $v_\tangent$ is located, for $i \neq j$. Let $C_k$ be the first circle hit by $v_\tangent$ during rotation at orientation~$\alpha$. If $k < i$ or $j < k$, then at orientation $\alpha$, $v_\tangent$ lies on the concave semi-circle of $C_k$ with orientation~$\alpha$.
\end{restatable}
\begin{proof}

    Without loss of generality assume $p_i$ is positioned left of $p_j$, then assuming that $k < i$ or $j < k$, $p_k$ must lie below the line through $p_i$ and $p_j$ (since $i < j$ and the points are ordered in clockwise order). See Figure~\ref{fig:between-ij} for the following construction.
    
    Consider point $p_{k'}$ placed on the line through $p_i$ and $p_j$, such that $v_\tangent$ could be the bottom intersection of~$C_j$ and a radius-$r$ circle centered at $p_{k'}$. Let $p_k$ be a point placed on the circle of radius $r$ centered at $v_\tangent$. If $p_k$ is placed on the other side of the line through $v_\tangent$ and $p_j$, compared to $p_{k'}$, then $v_\tangent$ would enter $C_k$ at orientation $\alpha$, which does not induce an event. As such, $p_k$ must be on the same side of the line through $v_\tangent$ and $p_j$ as $p_{k'}$. Additionally, since $k < i$ or $j < k$, $p_k$ must be below the line through $p_i$ and $p_j$.  

    \begin{figure}
        \centering
        \includegraphics{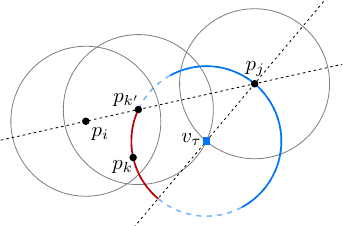}
        \caption{Point $p_k$ is placed on the red arc, which is a subset of the (blue dashed) convex semi-circle centered at $v_\tangent$, and disjoint from the (blue solid) concave semi-circle centered at~$v_\tangent$.}
        \label{fig:between-ij}
    \end{figure}

    It is easy to see that all points on the convex semi-circle of $v_\tangent$ will have $v_\tangent$ on their concave semi-circle. Since the arc on which $p_k$ is placed is a strict subset of the concave semi-circle of $v_\tangent$, any placement of $p_k$ must have $v_\tangent$ on its concave semi-circle.
\end{proof}

Lemma~\ref{lem:between_ij} implies that, while constructing a type-5 certificate, we need to consider only candidate points $p_k$ such that $i < k < j$. Note that, if $i=j$, we get a type-4 event. Then, all that is left is to find the first circle $C_k$ with $i < k < j$ that is intersected by $v_\tangent$ as it moves over $C_j$. To this end, we describe a data structure that allows us to perform a circular ray shooting query along $C_j$ from the orientation at which the certificate must be constructed.

\subparagraph{Data structure.}
Our data structure is essentially a balanced binary tree $\tree$ on the vertices of the convex hull in clockwise order, where each node stores an associated structure. For any~$i$, let $D_i$ be the disk bounded by circle $C_i$.
Suppose a node in $\tree$ is the root of a subtree with $p_i,\dots,p_j$ in the leaves. Then its associated structure stores the boundary of $\bigcap_{i\leq l\leq j} D_l$ as a sorted sequence of circular arcs. Given a range $(i, j)$ we can query this data structure with a (circular) ray in $O(\log^2 h)$ time to find the first intersection of the ray with the boundary of $\bigcap_{i\leq l\leq j} D_l$.
%For technical reasons we do not use one tree, but several trees based on a partition of the vertices. 
Section~\ref{sec:ds} describes the data structure and query algorithms in more detail. 

\subparagraph{Event handling.}

Whenever a certificate of type 5 is violated at orientation $\alpha$, it means some circle $C_k$ is hit by $v_\tangent$ at orientation $\alpha$. It is still possible, however, that this hit happens on the concave semi-circle of $C_k$. In that case, $C_k$ should not be added to $S_1$, and $v_\tangent$ should simply continue moving over its original trajectory. We call these events, where a type-5 certificate is violated but $S_1$ is not updated, \emph{internal events}. To handle an internal event, we merely need to construct another new type-5 certificate and continue our rotational sweep. In this case, however, we do not have to search the entire range $p_i, \ldots, p_j$ when constructing a new certificate, as shown in the following lemma.

\begin{restatable}{lemma}{constrainedQueryRange}
    \label{lem:constrained-query-range}
    Let $p_i$ be the point corresponding to $\tangent_1$, and let $v_\tangent$ be at the intersection of circles $C_j$ and $C_k$, where $C_j$ is the circle that defines the current trajectory of $v_\tangent$ and where $v_\tangent$ is on the concave semi-circle of $C_k$. Then if the next circle hit by $v_\tangent$ is $C_l$ for $i < l < k$, this circle is hit on its concave semi-circle.
\end{restatable}
\begin{proof}
    Let $C_l$ be the next circle hit by $v_\tangent$ for $i < l < k$, and see Figure~\ref{fig:IE-range} for the following construction. 
    Since $v_\tangent$ is on the concave semi-circle of $p_k$, $p_k$ must be on the convex semi-circle of $v_\tangent$. Furthermore, as $C_k$ was hit by $v_\tangent$, $p_k$ must lie on the same side of the line through $p_j$ and $v_\tangent$ as $p_i$. Let the endpoint of the concave semi-circle of $v_\tangent$ that lies clockwise from $p_k$ be denoted $v_c$, and observe that $d(p_k, p_j) > d(v_c, p_j)$, where $d(a, b)$ denotes the Euclidean distance between $a$ and $b$. Additionally, since we consider only points on the convex hull, point $p_l$ must lie above line $p_ip_k$ but below line $p_kp_j$, resulting in a cone with its apex at~$p_k$. 

    \begin{figure}
        \centering
        \includegraphics{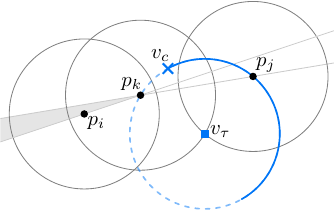}
        \caption{Point $p_l$ must be located in the gray cone. }
        \label{fig:IE-range}
    \end{figure}
    
    Every point in this cone is further away from $p_j$ than $p_k$: Since $v_\tangent$ is part of $S_1$, placing $q_1$ at $v_\tangent$ must yield a valid solution for $q_1q_2$. This means that all points must be in the Minkowski sum of a radius $r$ disc and the ray with orientation $\alpha$ originating from $v_\tangent$, and hence $p_k$ lies below the line $v_cp_i$ (see Figure~\ref{fig:internal-conjugate}). Finally, $p_j$ lies on the concave semi-circle around $v_\tangent$, and $p_k$ lies on the convex semi-circle around $v_\tangent$, which shows that the cone lies on the far side of $p_k$ with respect to~$p_j$. This implies that $d(p_l, p_j) > d(p_k, p_j)$. 
    
    During our monotone rotational sweep, $v_c$ must continuously move closer to $p_j$ as we continue our sweep. Therefore, when $v_\tangent$ intersects $C_l$, we must have $d(p_l, p_j) > d(v_c, p_j)$: This directly implies that $v_c$ must lie clockwise from $p_l$ on the radius $r$ circle centered at $v_\tangent$. Thus, $v_\tangent$ lies on the concave semi-circle of~$C_l$.
\end{proof}

Lemma~\ref{lem:constrained-query-range} implies that, after an internal event with circle $C_l$, it is sufficient to consider only points $p_k$ with $l < k < j$ when constructing a new type-5 certificate.

When a type-5 certificate is violated and $v_\tangent$ is on the convex semi-circle of $C_k$, however, we do need to update $S_1$ to reflect $v_\tangent$ moving over the intersection between $C_k$ and $C_j$. Additionally, we must construct new certificates: We possibly need to compute a new type-1 certificate, as well as either a type-4 certificate or a new type-5 certificate using $C_k$ as the new trajectory of $v_\tangent$ and searching for the next hit with $C_l$ for $i < l < k$. We are now ready to prove Lemmata~\ref{lem:event5-cert} and~\ref{lem:event5}.

\eventFivecert*
\begin{proof}
    Let $p_i$ be the point corresponding to $\tangent_1$, $C_j$ be the circle over which $v_\tangent$ is currently moving, and $\alpha$ be the current orientation. To construct a type-5 certificate, we find the first circle $C_k$ with $i \leq l < k < j$ that is intersected by $v_\tangent$. Here, if this certificate is constructed during an internal event, $l$ is the index of the circle $C_l$ intersected by $v_\tangent$ during that event. Otherwise, $l = i$. Since $v_\tangent$ moves over $C_j$, we can use the data structure described in Section~\ref{sec:ds} to perform a circular ray shooting query on the sequence $C_l, \ldots, C_{j-1}$ with starting point $v_\tangent$ to find $C_k$ in $O(\log^2 h)$ (Lemma~\ref{lem:query-circle}). This gives us the point $p_k$ to include in the certificate, and we can find the orientation at which $p_k$ is hit by drawing the tangent of $C_i$ through the intersection point between $C_j$ and $C_k$. 
\end{proof}

\eventFive*
\begin{proof}
    To show that this event happens $O(h \log h)$ times during a $\pi$ rotation, we need the following definition. Let an ordered pair $(p_g, p_h)$ of vertices of $\convexhull{P}$ be \emph{conjugate} if we can place a semi-circle of radius $r$ so that it hits $p_g$ and $p_h$, $p_h$ lies clockwise from $p_g$ on this semi-circle, and the semi-circle does not intersect the interior of $\convexhull{P}$. Efrat and Sharir prove that any convex polygon with $n$ vertices has at most $O(n \log n)$ conjugate pairs if the vertices are placed in general position~\cite{DBLP:journals/dcg/EfratS96}. We charge each occurrence of a type-5 event to a conjugate pair, and prove that each pair is charged only a constant number of times. This immediately yields the bound of $O(h \log h)$ on the number of type-5 events. 

   Consider an internal type-5 event. We charge these events to the pair $(p_k, p_j)$. To this end, we show that this pair of points must be conjugate. Consider orientation $\alpha$ at which this event takes place. At that point, we can draw a circle of radius $r$, centered at $v_\tangent$, through $p_k$ and $p_j$. Since, by definition of an internal event, $v_\tangent$ is on the concave semi-circle of $p_k$, $p_k$ must be on the convex semi-circle centered around $v_\tangent$.
   
   Now again observe that, since $v_\tangent$ is part of $S_1$, placing $q_1$ at $v_\tangent$ must yield a valid solution for $q_1q_2$. This means that all points must be in the Minkowski sum of a radius $r$ disc and the ray with orientation $\alpha$ originating from $v_\tangent$. See Figure~\ref{fig:internal-conjugate}. Therefore, the concave semi-circle of radius $r$ centered at $v_\tangent$ does not intersect $\convexhull{P}$. If we rotate this semi-circle counter-clockwise until one of its endpoints coincides with $p_k$, we obtain a semi-circle through $p_k$ and $p_j$ that does not intersect $\convexhull{P}$. This means it is a witness that $(p_k, p_j)$ is conjugate.

    \begin{figure}
        \centering
        \includegraphics{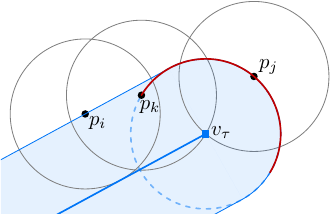}
        \caption{All points in $\convexhull{P}$ must be in the blue shaded area. When an internal event happens, the red semi-circle is a witness that $(p_k, p_j)$ is conjugate.}
        \label{fig:internal-conjugate}
    \end{figure}

   Next, consider a type-5 event that is not internal. The same argument as above holds, except $p_k$ is already on the concave semi-circle of radius $r$ centered at $v_\tangent$. Since this semi-circle does not intersect $\convexhull{P}$, that semi-circle is a direct witness that $(p_k, p_j)$ is conjugate. 

    Every conjugate pair only induces at most one type-5 event. In order for conjugate pair $(p_k, p_j)$ to induce a second type-5 event, $v_\tangent$ must again hit the intersection between $p_k$ and $p_j$ while it is moving in the same angular movement direction. This can only happen if we perform a full $2\pi$ rotational sweep, or if $v_\tangent$ first moves over this intersection in the opposite direction. The prior is not possible in a $\pi$ rotational sweep. The latter is only possible if $p_k$ is first involved in a type-4 event. But then, $k = i$, and by construction $p_k$ can never be involved in another type-5 certificate.

    Handling a type-5 event consists of updating $S_1$ and the movement trajectory of $v_\tangent$, which can be done in $O(1)$ time. Additionally, we construct new type-1 certificate and either a type-4 or type-5 certificate, which can be done in $O(1)$ and $O(\log^2 h)$ time by Observations~\ref{obs:event1-cert} and~\ref{obs:event4-cert} and Lemma~\ref{lem:event5-cert}.     
\end{proof}

After analyzing all events that occur during the rotational sweep, we can prove Theorem~\ref{thm:static-full}.

\staticFull*
\begin{proof}
    We initialize the algorithm by computing the convex hull $\convexhull{P}$. By Lemma~\ref{lem:convex-hull}, it is sufficient to consider only points in $\convexhull{P}$ to find a representative segment of $P$. We use rotating calipers to check that the shortest representative segment is not a point, and find an orientation~$\alpha$ in which a solution exists. For this fixed $\alpha$ we find $\tangent_1$ and $\tangent_2$, compute $\convexchain_1$ and $\convexchain_2$, as well as the shortest line segment $q_1q_2$ in orientation~$\alpha$. Computing the convex hull can be done in $O(n \log h)$~\cite{chan1996convexhull}. By Lemma~\ref{lem:sequence-construction}, $\convexchain_1$ and $\convexchain_2$ can be initialized in $O(h)$, and we can initialize $q_1q_2$ in $O(\log h)$ time. As such, initialization of the algorithm can be done in $O(n \log h)$ time in total.
       
    Next, we rotate orientation $\alpha$ over $\pi$ in total, maintaining $\tangent_1$, $\tangent_2$, $\convexchain_1$, and $\convexchain_2$, as well as $q_1q_2$. Note that a rotation of $\pi$ is sufficient, since we consider orientations, which identify opposite directions of a $2\pi$ rotation. Throughout the rotation we maintain the shortest representative segment, and return the shortest such line segment found over all orientations. Over the full rotation, we encounter five different types of events. By Lemmata~\ref{lem:event1}--\ref{lem:event5}, these events can be handled in $O(h \log^3 h)$ time in total. 

    As we mentioned in Section~\ref{sec:fixed-orientation}, the shortest representative segment is defined by only $\tangent_1$, $\tangent_2$, $\convexchain_1$, and $\convexchain_2$. Each tangent is defined by a circle. Hence, $\tangent_1$ and $\tangent_2$ can only change when they are defined by a new circle. Thus, by Lemma~\ref{lem:event2}, we correctly maintain $\tangent_1$ and $\tangent_2$ throughout the full $\pi$ rotation. Furthermore, $\convexchain_1$ and $\convexchain_2$ can only exist when $\tangent_1$ and $\tangent_2$ appear in the correct order. By Lemma~\ref{lem:event3} we correctly maintain when $\convexchain_1$ and $\convexchain_2$ exist. Finally, $\convexchain_1$ and $\convexchain_2$ can only make a discrete change as the tangents $\tangent_1$ and $\tangent_2$ hit intersections between circles in $\mathcal{C}_P$. If the tangents do not touch a vertex, then $\convexchain_1$ and $\convexchain_2$ must change continuously along the arcs that $\tangent_1$ and $\tangent_2$ cross. By Lemmata~\ref{lem:event4} and~\ref{lem:event5}, we correctly maintain $\convexchain_1$ and $\convexchain_2$ when they exist.
    In conclusion, we correctly maintain $\tangent_1$, $\tangent_2$, $\convexchain_1$, and $\convexchain_2$ throughout the full $\pi$ rotation. Since we maintain the shortest line segment between $\convexchain_1$, and $\convexchain_2$ in any orientation using Lemma~\ref{lem:event1}, we also maintain the shortest representative segment for any orientation. 
\end{proof}

\subsection{Data structure and queries}
\label{sec:ds}

We describe a data structure that stores the points of $\convexhull{P}$, so that for a given query subsequence $p_i,\ldots,p_j$ and another object $q$ (a radius-$r$ circle or a line),
we can efficiently find the two intersections of $\bigcap_{i\leq l\leq j}D_l$ with $q$, or decide that they do not intersect. Here, $D_l$ is the disk of radius $r$ centered on $p_l$. 
Since the common intersection of disks is convex, a line will intersect it at most twice.
Furthermore, since the query circle has the same radius as each of the disks $D_l$, it will intersect the common intersection of any subset of these disks at most twice as well. 
In all queries we will perform, any pair of points in the query sequence $p_i,\ldots,p_j$ are at most $2r$ apart. Furthermore, we are always interested in the intersection that is the \emph{exit}, when we direct the line or when we direct a circular ray along the circle from a given starting point. Computing both intersections of the line or circle immediately gives this desired exit intersection as well.

We first break the sequence $p_1, \ldots, p_h$ of $\convexhull{P}$ into a number of subsequences with useful properties, see Figure~\ref{fig:partit}. We will create maximal subsequences $p_s,\ldots,p_t$ so that the total turning angle is at most $\pi/2$, and the Euclidean distance between $p_s$ and $p_t$ is at most $r$. These two properties ensure that if we place radius-$r$ disks centered on the points of $p_s,\ldots,p_t$, their common intersection is non-empty. Maximality of the subsequences ensures that for any query, we need to search in only $O(1)$ subsequences. The construction of such subsequences can be done in an incremental greedy manner and is standard. We can assume that there are at least two points in $p_1,\ldots,p_h$ that are more than $r$ apart, otherwise all points fit in a radius-$r$ circle.

Let $p_s,\ldots,p_t$ be any one of these subsequences, let $C_s,\ldots,C_t$ be the radius-$r$ disks centered on these points, and let $\MC_{st}$ be their common intersection. We first describe a data structure for $\MC_{st}$, which will later be used as an associated structure for $p_s,\ldots,p_t$.

\begin{figure}
    \centering
    \includegraphics{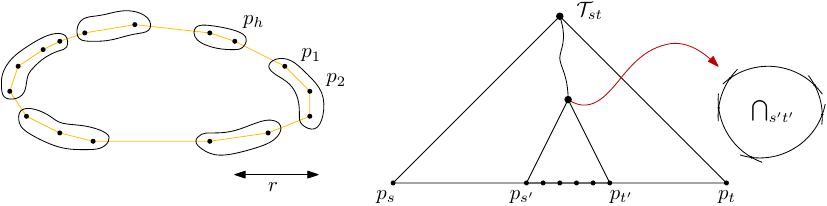}
    \caption{Left, partitioning $p_1,\ldots,p_h$ into subsequences by distance and angle. Right, schematic view of the data structure $\tree_{st}$.}
    \label{fig:partit}
\end{figure}

The queries on $\MC_{st}$ that we need to answer are the following: (1) given a circle $C$ of radius $r$ and a point $q$ on it, where $q$ lies inside $\MC_{st}$, report the first circular arc hit when $q$ moves counterclockwise along $C$. (2) Given a query line $\ell$, find the intersections with $\MC_{st}$. It is clear that the second query yields at most two points. For the first query, since $C$ and $C_s,\ldots,C_t$ all have the same radius, $C$ intersects $\MC_{st}$ in at most two points as well.
%In other words, one subsequence of the vertices of $\MC_{st}$ lies inside $C$ and the other subsequence lies outside $C$. 
These properties allow us to store the vertices of $\MC_{st}$ in sorted order and perform a simple binary search with a radius-$r$ query circle or a line, and find the intersected arcs in $O(\log h)$ time.
%to find the answer to the circular ray shooting query from $q$ on $C$.

The data structure for the subsequence $p_s,\ldots,p_t$ is a balanced binary search tree $\tree_{st}$ that stores $p_s,\ldots,p_t$ in the leaves, see Figure~\ref{fig:partit}. For every internal node, with leaves containing $p_{s'}, \ldots,p_{t'}$ ($s\leq s'\leq t'\leq t$), we store as an associated structure the vertices of $\MC_{s't'}$ in sorted order, as just described. The complete data structure $\tree$ is the set of all of these trees. Breaking up into subsequences guarantees that every node in any tree in $\tree$ has a non-empty associated structure.

It can easily be seen that the data structure has size $O(h\log h)$ and can be constructed in $O(h\log^2 h)$ time (even $O(h\log h)$ construction time is possible).

\subparagraph{The circular query.}
Suppose we wish to perform a query on the vertices $p_i,\ldots,p_{j-1}$ with a circle $C$ and starting point $q$. Since $p_i,\ldots,p_{j-1}$ is part of the cyclic sequence $p_1,\ldots,p_h$ of $\convexhull{P}$, it may be that $p_h$ and $p_1$ are in the sequence $p_i,\ldots,p_{j-1}$. In that case,
instead of querying with $p_i,\ldots,p_{j-1}$, we query with the two sequences $p_i,\ldots,p_{h}$ and $p_1,\ldots,p_{j-1}$ instead. For ease of description we assume that $p_h$ and $p_1$ are not part of $p_i,\ldots,p_{j-1}$, so $j-1\geq i$.

Let $c$ be the center of $C$; in our algorithm, $c$ will always be the convex hull vertex $p_j$. 
We first test if $p_i$ is farther than $2r$ from $p_j$: if so, we binary search for the vertex $p_{i'}$ with $i<i'\leq j$ that is at most $2r$ from $p_j$ and set $i=i'$. Now we are guaranteed that $p_i$ is within distance $2r$ from $p_j$. Points farther than $2r$ from $p_j$ cannot be answers to the query.
%The queries our algorithm needs to support have an additional property: $p_i,\ldots,p_{j-1}$ makes a total turn of less than $\pi$ (in clockwise direction).

We find the subsequences containing $p_i$ and $p_{j-1}$; they may be the same, they may be adjacent, or there may be $O(1)$ subsequences in between. It cannot be more because that would violate the greedy choices of the subsequences.

If $p_i$ and $p_{j-1}$ lie in the same subsequence $p_s,\ldots,p_t$, we use the search paths in $\tree_{st}$ to $p_i$ and to $p_{j-1}$ and find the highest subtrees between these paths plus the leaves containing $p_i$ and $p_{j-1}$, see Figure~\ref{fig:ds} (right). At the roots of these subtrees, we query the associated structure as described before. At the leaves, we check the one circle.

If $p_i$ and $p_{j-1}$ lie in two adjacent subsequences, we follow the path to $p_i$ in the tree that contains it, and collect the highest nodes strictly right of the search path, where we query the associated structures of their roots. Then we follow the path to $p_{j-1}$ in the tree for the other subsequence and collect the highest nodes strictly left of the search path, where we query the associated structures of their roots. Again we include the leaves with $p_i$ and with~$p_{j-1}$.

If there were subsequences in between, we additionally query the associated structure of the root of these trees.

In all cases, we choose the best answer (first hit) among the answers from all of the $O(1)$ subsequences and all of the $O(\log h)$ subtrees.

\begin{lemma}
    \label{lem:query-circle}
    A circular query in $\tree$ can be answered in $O(\log^2h)$ time.
\end{lemma}

\begin{figure}
    \centering
    \includegraphics{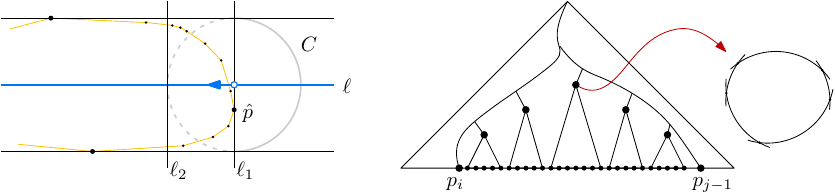}
    \caption{Left, the situation when the leftmost right cap is computed. Right, a query in the data structure shown schematically, with a set of highest subtrees between the search paths to $p_i$ and $p_{j-1}$.}
    \label{fig:ds}
\end{figure}

\subparagraph{The line query.}
When at a certain angle $\alpha$, a solution starts to exist again (type-3 event), we must find this first solution so that we can continue the rotating calipers process. At angle $\alpha$ we have two parallel lines that enclose all points of $\convexhull{P}$, there is one point on each line, and we need to find the caps of the hippodrome: the semi-circles of radius $r$ with their endpoints on the two lines, such that all points of $P$ lie in the hippodrome, and the semi-circles are shifted inwards as much as possible (until they hit a point of $P$). For ease of description we assume that $\alpha=0$, so the parallel lines are horizontal. We consider where to place the cap that bounds the hippodrome from the right, moving it as far left as possible.

We first find the rightmost point in $\convexhull{P}$; let it be $\hat{p}$ and let its $x$-coordinate be $\hat{p}_x$; see Figure~\ref{fig:ds} for an illustration.
Let $\ell_1\,:\,x=\hat{p}_x$ be the vertical line through the rightmost point, and let  $\ell_2\,:\,x=\hat{p}_x-r$. Let $\ell$ be the horizontal line midway between the two parallel lines that enclose the points.
We need to find a point $q$ on $\ell$, such that the radius-$r$ circle centered at $q$ has all points of $P$ to the left of (or on) its right semi-circle, and that right semi-circle contains one point of $P$: the answer to the query. Note that $q$ will lie between $\ell_2$ and $\ell_1$. Consequently, the answer to the query is a point to the right of $\ell_2$. Therefore, we need at most four subsequences from the partition of $\convexhull{P}$, and in the following, we simply query each of the corresponding trees to find the answer. For simplicity we refer to any such tree as $\tree$.

We will find our answer using three separate queries. The first involves all points of $P$ that lie inside $C$, the radius-$r$ circle that has the parallel lines as tangents and whose center is $\ell\cap\ell_1$. Point $\hat{p}$ is one of the points inside $C$; in general the points inside $C$ form a subsequence of $\convexhull{P}$  that includes $\hat{p}$. If points beyond the two line points (on the parallel horizontal lines) lie inside $C$, we ignore them since they are not the desired answer (this may happen when line points lie right of $\ell_2$). 
%The line points themselves, however, can be the answer if they lie right of $\ell_2$.

We can find the two relevant intersections of $\convexhull{P}$ and $C$ in $O(\log h )$ time by binary search. For this sequence $p_i,\ldots,\hat{p},\ldots,p_j$ of points, we know that the circles centered at them all contain the center of $C$. So, a line intersection query with $\ell$ in their common intersection gives two intersection points. The left one is the answer we want, and we report the point whose left side of the circle was intersected rightmost. We query the data structure $\tree$ with $p_i$, $p_j$, and $\ell$ for this. As before, we in fact query multiple trees and subtrees, and we report the rightmost one of the different answers.

The other two queries are handled in the same way, so we describe only one of them. It concerns the points of $\convexhull{P}$ that lie outside $C$, to the right of $\ell_2$, and before $\hat{p}$ in the clockwise order. 
%Recall that the first point of interest is the point of $\convexhull{P}$ is the one on the top parallel line.
Let the sequence of these points be $p_i,\ldots,p_j$ (overloading notation somewhat). We can find $p_i$ and $p_j$ by binary search on $\convexhull{P}$ using $\ell_2$ and $C$. The sequence $p_i,\ldots,p_j$ has useful properties, when considering the radius-$r$ disks $D_i,\ldots,D_j$ centered on them. 
All of these disks intersect the line $\ell$; let these intervals be $I_i,\ldots,I_j$. The length of intersection increases monotonically from $i$ to $j$.  The right endpoints of the intervals appear in the order $i,\ldots,j$: the right endpoint of $I_l$ is always left of the right endpoint of $I_{l+1}$. The left endpoints are not sorted.

The common intersection $\bigcap_{i\leq l\leq j} D_l$ is not empty, but this common intersection may not intersect $\ell$. If it does intersect $\ell$, then we can simply find the two intersection points and the left one is our answer (the point whose disk contributes to the common intersection with the circular arc that is hit).
We can find it using the data structure $\tree$ by querying with $p_i$, $p_j$, and $\ell$.

When $\bigcap D_l\cap \ell=\emptyset$, then this does not work. Fortunately, we can still use $\tree$, but the query algorithm becomes more involved.
For $i\leq i'\leq j'\leq j$, let $\MC_{i'j'}$ denote the common intersection of the disks $D_{i'},\ldots,D_{j'}$.

\begin{figure}
    \centering
    \includegraphics{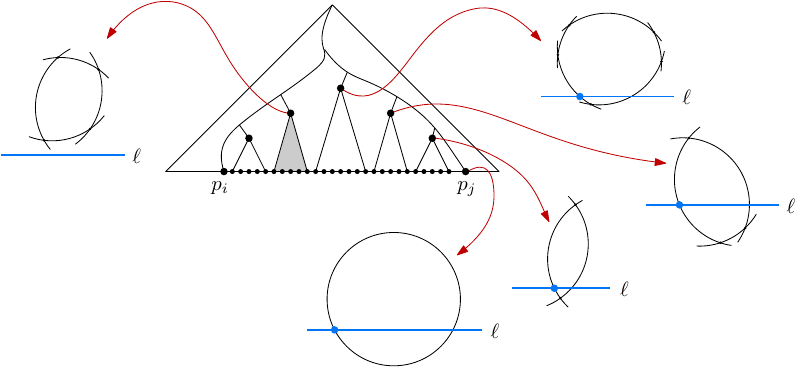}
    \caption{The query in $\tree$ among $p_i,\ldots,p_j$ with line $\ell$. In the associated structures, from right to left, we find a non-empty intersection four times (red points are answers in these subtrees), and for the fifth node, $\ell$ does not intersect the common intersection: we must descend into the grey~subtree.}
    \label{fig:linequery}
\end{figure}

We query $\tree$ with $i$, $j$, and $\ell$. The query with $i$ and $j$ results in $O(\log h)$ subtrees that together represent all points $p_i,\ldots,p_j$, see Figure~\ref{fig:ds}. Each root of such a subtree stores the non-empty common intersection of the disks corresponding to the points in the leaves of that subtree, as stated in the description of $\tree$. We treat these subtrees from higher indices towards the ones with lower indices, so first the subtree for which $p_j$ is the rightmost point. See Figure~\ref{fig:linequery}.
Suppose we have treated some subtrees and now we need to handle a subtree $T_{i'j'}$ containing the points $p_{i'},\ldots,p_{j'}$, where $i\leq i'\leq j'\leq j$. The root of this subtree stores $\MC_{i'j'}$. We compute $\MC_{i'j'}\cap \ell$ in $O(\log h)$ time. Now there are two options: (1) If $\MC_{i'j'}\cap \ell=\emptyset$, then the answer to the original query cannot be in $p_i,\ldots,p_{i'-1}$, so we do not need to query any more subtrees that store points with smaller indices than the present one. However, the answer may be in this subtree itself, so we must descend in it. (2) If $\MC_{i'j'}\cap \ell\not =\emptyset$, then the answer among $p_{i'},\ldots,p_{j'}$ is the point whose disk provides the arc of the left intersection point of $\MC_{i'j'}\cap \ell=\emptyset$, which implies that we do not need to descend in this subtree. We already have the answer.

So for the selected subtrees from right to left, we first get a sequence with a non-empty intersection of their common intersection $\MC$ with $\ell$, and then one where the intersection is empty, as illustrated in Figure~\ref{fig:linequery} (we ignore subtrees further to the left).
We need to search in this subtree, and do that as follows. We first query the right subtree of the root, and again perform the test whether the common intersection in the right subtree intersected with $\ell$ is empty or not. If it is empty, then we can ignore the left subtree and by the same argument as before, and we continue the same way in the right subtree until we reach a leaf. 

The correctness of this search hinges on the claim that if a subtree gives an empty intersection with $\ell$, then no subtree more to the left---with lower indices---can give the answer. The argument is simple, when thinking about the intervals $I_i,\ldots,I_j$. If the intersection for a subtree is empty, we necessarily have a left endpoint among the sequence of right endpoints in this subtree. This left endpoint gives a candidate answer. So, intervals more to the left for which we still need to encounter the right endpoint (reasoning still from right to left) certainly cannot yield a better left endpoint than the one that we know to exist.

The query visits $O(\log h)$ nodes in $\tree$, namely the nodes of up to two search paths to leaves and their direct children next to these paths. At every node we spend $O(\log h)$ time for the intersection query of the common intersection with $\ell$.

\begin{lemma}
    \label{lem:query-line}
    A line query in $\tree$ can be answered in $O(\log^2h)$ time.
\end{lemma}

\section{Stable Representative Segments for Moving Points}\label{sec:moving points}

In this section, we consider maintaining representative segment $q_1q_2$ while the points in $P$ move. We first show what can happen if we would maintain the optimal solution explicitly under continuous motion of the points.

There are examples where a representative segment exists for a value of $r$ for an arbitrarily short duration. This can happen because one point moves towards a hippodrome shape and another point moves out of it, giving a brief moment with a valid hippodrome. The same effect can be caused by a single linearly moving point that grazes the hippodrome at a join point. These examples are illustrated in Figure~\ref{fig:kinetic-ex}(a).
It can also happen that a minor oscillating movement of a point causes a quick sequence of changes between a valid and no valid segment, see Figure~\ref{fig:kinetic-ex}(b).

\begin{figure}
    \centering
    \includegraphics{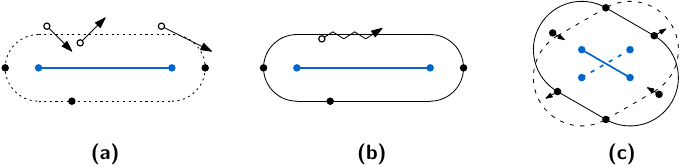}
    \caption{Examples of representative segments (blue) for moving points. \textbf{\textsf{(a)}} shows that a representative segment may flicker arbitrarily briefly due to the motion of two points (top left) or one point (top right). \textbf{\textsf{(b)}} shows that the representative segment may flicker due to minor oscillations in a movement path. \textbf{\textsf{(c)}} Small movements can trigger a discrete change of the representative segment.}
    \label{fig:kinetic-ex}
\end{figure}

Now, consider the point set~$P$ in which the points form a regular $k$-gon (see Figure~\ref{fig:kinetic-ex}(c)). When $r$ is equal to half the width of the $k$-gon, then we can force a discrete change in the placement of $q_1$ and $q_2$, with very slow movement of points in~$P$. In this case there is always a valid representative segment, but its endpoints make a jump.

We see that continuously maintaining the optimal solution has two undesirable artifacts: (1) the existence and non-existence of a solution can be arbitrarily short, leading flicker in a visualization, and (2) a solution may jump to a new solution, leading to infinitely high speeds of the endpoints of the segment even if the points themselves move slowly.

Our goal is therefore to ensure that the segment movement is \emph{stable} over time: We want $q_1$ and $q_2$ to move continuously and with bounded speed, preventing discrete or (near) instantaneous changes. We accomplish this as follows. First, we will sample the positions of the moving points only at integer moments, and then decide immediately to show or not show a solution in the next time unit. This implies that existence and non-existence of a solution lasts at least one time unit, and we avoid flickering.
Second, we make the assumption that the maximum speed of the points is unit. We have to make some bounded speed assumption otherwise we cannot hope to get a bounded speed of the endpoints either.
Third, we allow more flexibility in when we have a solution. We do this using a less strict regime on $r$, and by assuming that $r\geq 1$. Whenever the real solution exists (with the actual $r$), then we guarantee that we also have a solution. Whenever there is no solution even for a radius of $2\sqrt{2}\cdot r +4$, then we never give a solution. When the radius is in between these bounds, we may have a solution or not.
Our algorithm can then ensure that the speeds of the endpoints are bounded, and the length of our chosen segment always approximates the true optimum (when it exists), at any moment in time, also between the integer sampling moments.

So we assume that a point $p \in P$ is described by a trajectory that is sampled at integer timestamps. That is, each point in $P$ is described by $p(i) \rightarrow \mathbb{R}^2$ where $i \in \mathbb{Z}$ is the timestamp. 
We also assign each endpoint of the representative segment a position as a function of $t\in \mathbb{R}$, which means that the representative segment at time $t$ is now defined by $q_1(t)q_2(t)$. Furthermore, we use $D(t)$ and $W(t)$ as the \emph{diameter} and \emph{width} of the point set, which are respectively the maximum pairwise distance of points in $P$ and the width of the thinnest strip containing~$P$. Let $d_1(t),d_2(t)\in P$ be a pair of points defining the diameter. Lastly, we also use the \emph{extent} $E(t)$ of $P$ in the direction orthogonal to the line segment $d_1(t) d_2(t)$. For all of the above definitions, we omit the dependence on $t$ when it is clear from the context.

In the following, we describe how to specify $q_1(t)$ and $q_2(t)$ such that they move with bounded speed, and such that the length of the segment $q_1(t)q_2(t)$ as well as the proximity of the segment to $P$ can be bounded at any time $t$. In particular, we define such a segment $q_1(t)q_2(t)$ as an \emph{approximating} segment, and prove that the length of an optimal representative segment is approximated by an additive term~$l$, and at the same time the maximum distance from any point in $P$ to the segment $q_1(t)q_2(t)$ is at most $h\cdot r$, for some constant~$h$. 
%These assumptions influence the values $l$ and $h$ as well as the maximum speed of $q_1(t)$ and $q_2(t)$; there are options for trade-offs. %\jules{I think we are essentially creating a trade-off between $r$, the movement speed of the points, and our (integer) time steps: In most of our discussions we use that, in one time step, our points move at most $r$. This leads to the following trade-offs: If we increase $r$ (worse solution quality and faster moving segment), our time steps can increase (less flickering). Or otherwise, if points move faster, we have to shorten the time steps to ensure $r$ movement per time step, and thus we get more flickering. We could also leave the same time steps, but increase~$r$, which results in worse solution quality and a faster moving segment. It makes sense that in erratic inputs (faster moving points) our quality metrics deteriorate, and this shows we get to choose which part of the trade-off takes the hit.} \marc{As there are several interrelated values ($h$, speed, sampling rate), we can use "for ease of description we assume that the maximum speed of points is $1$, we sample at every integer time step, and $r\geq 1$. These assumptions influence the values $l$ and $h$ as well as the maximum speed of $q_1(t)$ and $q_2(t)$; there are options for trade-offs."}

\begin{figure}
    \centering
    \includegraphics{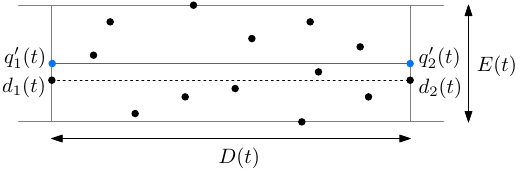}
    \caption{The prospective segment $q_1' q_2'$ in relation to $D$, $E$, and diametrical line $d_1 d_2$ at time~$t$.}
    \label{fig:kinetic-def}
\end{figure}

\subparagraph{Algorithm.} Our algorithm $A(t)$ is \emph{state-aware}. This means that the output of $A(t)$ is dependent only on the input at or before time $t$, but it has no knowledge of the input after time $t$. At every integer timestamp $i \in \mathbb{Z}$, we compute a \emph{canonical solution} $q'_1(i)q'_2(i)$. The endpoints $q'_1(i)$ and $q'_2(i)$ of this canonical solution are placed on the lines orthogonal to the diametric line through $d_1(i)$ and $d_2(i)$, respectively, such that $q'_1(i)q'_2(i)$ lies in the middle of the narrowest strip containing $P$ in the diametric orientation, see Figure~\ref{fig:kinetic-def}.

Our algorithm is now a special kind of state-aware algorithm called a \emph{chasing algorithm}~\cite{meulemans2019stability}: At at every integer time step $i \in \mathbb{Z}$, the algorithm computes the canonical solution $q'_1(i)q'_2(i)$ and then linearly interpolates from $q'_1(i-1)q'_2(i-1)$ to $q'_1(i)q'_2(i)$, arriving there at time $i+1$. At that point, $q'_1(i+1)q'_2(i+1)$ can be computed, and we continue in a similar manner.

However, if the maximum distance from $P$ to the canonical solution becomes too large, we no longer want the algorithm to output any solution. In this case, no (optimal) representative solution exists, and we do not produce an approximating segment either.

Formally, for any timestamp $t \in (i, i+1)$, we linearly interpolate $q_1$ and $q_2$ between their previous canonical placements as follows. We define $\alpha = (t-i)$ and set
$q_j(t) = \alpha \cdot q'_j(i-1) + (1 - \alpha) \cdot q'_j(i),\mbox{ for $j\in\{1,2\}$.}$
%\[q_2(t) = \alpha \cdot q'_2(i-1) + (1 - \alpha) \cdot q'_2(i).\]
Then, the output of our algorithm is $A(t)= 
    q_1(t)q_2(t)$ if  $E(\lfloor t \rfloor) \leq 2r\sqrt{2} + 2$ and $
    \varnothing$ otherwise.

%The first observation we can make is that the segment computed by our algorithm prevents flickering between states with or without a solution. This contrasts our earlier observation that the (optimal) representative segment can flicker between these states.
%\begin{observation}
%    Algorithm $A$ swaps between $q_1q_2$ and $\varnothing$ only at integer time steps.
%\end{observation}

The above algorithm yields a number of bounds, on the values of $W(t)$ for which algorithm $A$ outputs $q_1q_2$/$\varnothing$, the speed of the endpoints $q_1$ and $q_2$, the length of $q_1q_2$, and the distance to $P$ separately. We prove these bounds in Lemmata~\ref{lem:exist-lowerbound}--\ref{lem:hausdorff-bound}. First, we relate the extent $E$ of $P$ to its width $W$ in the following lemma.

\begin{restatable}{lemma}{widthExtent}
    \label{lem:width-extent}
    At any time $t$, $W(t) \leq E(t) \leq W(t)\sqrt{2}$.
\end{restatable}
\begin{proof}
    Consider a diametrical pair $d_1,d_2\in P$ at some time~$t$. The line segment $d_1d_2$ has length~$D$, and the extent orthogonal to the line segment has length~$E$. First observe that the width $W$ of~$P$ at time~$t$ is at most~$E$, as the thinnest strip in the orientation of $d_1d_2$ has width~$E$. To complete the proof, we show an upper bound of~$\sqrt{2}$ on the ratio~$E/W$. Note that this ratio is maximized when $W$ is as low as possible with respect to~$E$.

    To prove the upper bound, we first consider the case where~$P$ is rather symmetric, and later show that in other cases the proved upper bound cannot be exceeded. For the symmetric case, see Figure~\ref{fig:width-extent}(a). Let $d_1d_2$ be centered in the thinnest strip in the orientation of~$d_1d_2$. Since this strip cannot be thinner, there must be at least one point located on each line bordering the strip. Since the diameter of~$P$ is~$D$, no two points can be further than~$D$ apart. This limits the position of the points~$p_l,p_r$ on the boundaries of the strip, with respect to~$d_1,d_2$ and each other. To minimize the width of~$P$, we place $p_r$ closer to~$d_1$ and $p_l$ closer to~$d_2$, symmetrically, such that the line segment $p_lp_r$ also has length~$D$. As a result, $d_1d_2$ and $p_lp_r$ bisect each other and the angle~$\alpha$ between $d_1d_2$ and $p_lp_r$ is equal to $\arcsin{\frac{E}{D}}$.

    To compute the width~$W$ of~$P$ in this case, observe that the lines through~$d_1$ and~$p_l$ and through~$d_2$ and~$p_r$ are parallel (because of symmetry). We can now compute the width~$W$ as the (orthogonal) distance between those parallel lines.\\
    \begin{align*}
        W &= D\cdot \sin(\alpha / 2) & \\
        &= D\cdot \sqrt{\frac{1-\cos{\alpha}}{2}} & \text{(half-angle identity)}\\
        &= D\cdot \sqrt{\frac{1-\sqrt{1-\sin^2{\alpha}}}{2}}  & \text{(Pythagoras)}\\
        &= D\cdot \sqrt{\frac{1-\sqrt{1-(E/D)^2}}{2}}  & (\arcsin{\frac{E}{D}})\\
        &= \sqrt{\frac{D^2-\sqrt{D^2-E^2}}{2}}  &
    \end{align*}
    By definition $E\leq D$, thus we maximize~$E/W$ by setting $E=D$, resulting in $E/W=\sqrt{2}$.

    Finally, consider a non-symmetric case: $d_1d_2$ is not centered in the strip of width~$E$ and/or $p_l$ and $p_r$ are not placed symmetrically. Still $d_1$ and $d_2$ are horizontally aligned, to be exactly distance~$D$ apart, and $p_r$ is closer to~$d_1$, while $p_l$ is closer to~$d_2$ (see Figure~\ref{fig:width-extent}(b)). To minimize the width, $p_l$ and $p_r$ are also as far apart as possible, and hence exactly at distance~$D$ from each other. Since $d_1d_2$ and $p_lp_r$ no longer bisect each other, the strip of the symmetric case can no longer contain all the points: If the strip is centered on $p_lp_r$, then either $d_1$ or $d_2$ will be outside the strip, as in Figure~\ref{fig:width-extent}(b), and an analogous reasoning holds for $d_1d_2$. Finally, the strip of the symmetric case fits both $d_1d_2$ and $p_lp_r$ individually, but a strip in any other direction will have to accommodate one of both line pieces in a more orthogonal orientation, and thus has to be wider.
\end{proof}

\begin{figure}
    \centering
    \includegraphics{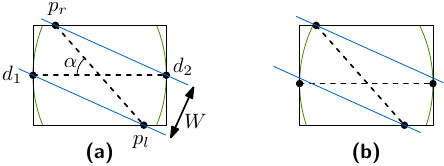}
    \caption{Constructions for Lemma~\ref{fig:width-extent}. Segments $d_1d_2$ and $p_lp_r$ are \textbf{\textsf{(a)}} symmetric or \textbf{\textsf{(b)}} off-center with respect to this bounding box. The (blue) thinnest strip in \textbf{\textsf{(a)}} is also shown in~\textbf{\textsf{(b)}}.}
    \label{fig:width-extent}
\end{figure}

We can now prove for which values of $W(t)$ algorithm $A$ is guaranteed to output $q_1q_2$, and when it is guaranteed to output $\varnothing$.

\begin{restatable}{lemma}{existLowerBound}
\label{lem:exist-lowerbound}
     At any time $t$, if $W(t) \leq 2r$ then $A(t) = q_1(t)q_2(t)$.
\end{restatable}
\begin{proof}
    Let $t \in [i, i+1]$ for some $i \in \mathbb{Z}$. For the sake of contradiction, assume $W(t) \leq 2r$ but $A(t) = \varnothing$. By definition, this means that $E(\lfloor t \rfloor) > 2r\sqrt{2} + 2$. Since the points in $P$ have at most unit speed, observe that $E$ can change by at most $2$ per time unit. As such, we get 
    \begin{equation*}
        2r\sqrt{2} < E(\lfloor t \rfloor) - 2 \leq E(t)
    \end{equation*}
    But then, by Lemma~\ref{lem:width-extent}, we have
    \begin{equation*}
        2r\sqrt{2} < E(t) \leq W(t)\sqrt{2}
    \end{equation*}
    which implies $2r < W(t)$. We have reached a contradiction, and the lemma follows.
\end{proof}

\begin{restatable}{lemma}{existUpperBound}
    \label{lem:exist-upperbound}
    At any time $t$, if $W(t) > 2r\sqrt{2} + 4$ then $A(t) = \varnothing$. 
\end{restatable}
\begin{proof}
    Let $t \in [i, i+1]$ for some $i \in \mathbb{Z}$. For the sake of contradiction, assume $W(t) > 2r\sqrt{2}+4$ but $A(t) = q_1(t)q_2(t)$. By definition, this means that $E(\lfloor t \rfloor) \leq 2r\sqrt{2}+2$. Since the points in $P$ have at most unit speed, $E$ can change by at most 2 per time unit. As such, we get 
    \begin{equation*}
        E(t) \leq E(\lfloor t \rfloor) + 2 \leq 2r\sqrt{2}+4
    \end{equation*}
    But then, by Lemma~\ref{lem:width-extent}, we have
    \begin{equation*}
        W(t) \leq E(t) \leq 2r\sqrt{2}+4
    \end{equation*}
    We have reached a contradiction, and the lemma follows.
\end{proof}

Next, we prove an upper bound on the speed on $q_1$ and $q_2$, as well as on the length of $q_1q_2$ and its distance to $P$.

\begin{restatable}{lemma}{speedBound}
    \label{lem:speed-bound}
    While $A(t) = q_1(t)q_2(t)$, endpoints $q_1$ and $q_2$ move continuously with a speed bounded by $(2r+1)\sqrt{2}+2$.
\end{restatable}
\begin{proof}
    To bound the speed with which $q_1$ (and symmetrically $q_2$) moves, we bound the distance between $q_1'(i)$ and $=q_1'(i+1)$, for $i\in\mathbb{Z}$. First observe that the diametric pair of points in~$P$ determining the orientation of $q_1'(i)q_2'(i)$ may not be the diametric pair at time~$i+1$, and this new diametric pair can therefore be located anywhere in the width~$E$ strip in the orientation of~$q_1'(i)q_2'(i)$. In a single time step, all point can move at most one unit outside this strip, and past the lines orthogonal to this strip through $q_1'(i)$ and $q_2'(i)$ (see Figure~\ref{fig:kinetic-speed}a). We bound the distance between~$q_1'(i)$ and~$q_1'(i+1)$ by half the lengths of the minor axes (remember, $q_1'(i)$ and~$q_1'(i+1)$ are located in the middle of the respective minor axes) plus the distance between the minor axes of the bounding boxes in the orientation of~$q_1'(i)q_2'(i)$ and~$q_1'(i+1)q_2'(i+1)$.
    
    Let~$p$ be a point at time~$i$ that is located on the minor axis of the time-$i$ bounding box, and consider the bounding box at time~$i+1$ (in the direction of $q_1'(i+1)q_2'(i+1)$). A point must be located on each axis of the bounding box at time~$t+1$, as otherwise it would not be a bounding box. Thus, each axis must have at least one point at distance at most one from the bounding box at time~$t$. All points lie in this bounding box at time~$i+1$, so the point~$p$ has to move at least to the boundary of the box at time~$t+1$. If it reaches the minor axis, the distance between the minor axes is one, so assume $p$ reaches only the major axis of the box at time~$t+1$ (see Figure~\ref{fig:kinetic-speed}b). Since the minor axis at time~$t+1$ has to be at distance at most one from the boundary of the box at time~$t$, it can be at distance at most~$\sqrt{2}$ from~$p$.

    \begin{figure}
        \centering
        \includegraphics{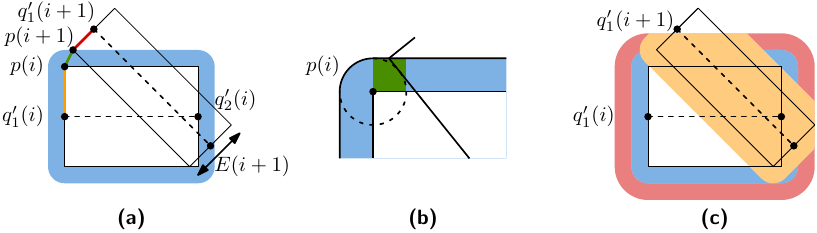}
        \caption{Constructions for Lemma~\ref{lem:speed-bound} and~\ref{lem:hausdorff-bound}. \textbf{\textsf{(a)}} All points are inside the (horizontal) bounding box at time~$i$ and hence cannot have moved outside the blue border at time~$i+1$. We bound the length of the three colored line pieces to bound the speed of~$q_1$. \textbf{\textsf{(b)}} When a point~$p$ on the minor axis of the horizontal box reaches the major axis of the box at~$i+1$, the minor axis of the same box must touch the green square. \textbf{\textsf{(c)}} At time~$i+2$ all points cannot have moves outside the width-2 red border around the box at time~$i$, and at most one unit outside the box at time~$i+1$ (yellow).}
        \label{fig:kinetic-speed}
    \end{figure}

    Finally, we add up all the distances. When both $q_1'(i)$ and~$q_1'(i+1)$ are used as outputs of our algorithm, neither~$A(i+1)$ nor~$A(i+2)$ is $\varnothing$. Thus, the minor axes are upper bounded by $E\leq 2r\sqrt{2}+2$, and the distance between~$q_1'(i)$ and~$q_1'(i+1)$ is upper bounded by~$2*\frac{E}{2}+\sqrt{2}=(2r+1)\sqrt{2}+2$.
\end{proof}

\begin{figure}
    \centering
    \includegraphics{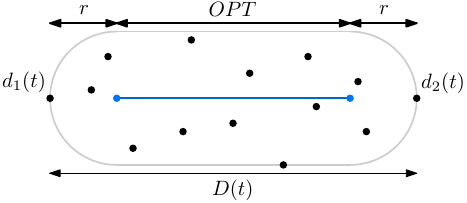}
    \caption{An example of the shortest possible representative segment.}
    \label{fig:OPT2r}
\end{figure}

\begin{restatable}{lemma}{lengthBound}
    \label{lem:length-bound}
    At any time $t$, if there exists a representative segment of $P$ with radius $r$ and $A(t) = q_1(t)q_2(t)$ then $|q_1(t)q_2(t)| \leq \OPT(t) + 2r + 4$, where $\OPT(t)$ is the length of the shortest representative segment of $P$ with radius $r$ at time $t$.
\end{restatable}
\begin{proof}
    We first argue that the canonical solution at time $i \in \mathbb{Z}$ has length at most $\OPT(i) + 2r + 2$. Observe that the length of $q'_1(i)q'_2(i)$ is $D(i)$. The shortest possible segment that is at distance at most $r$ from both diametrical points $d_1(i)$ and $d_2(i)$ has length $D(i) - 2r$, and has its each endpoint at distance exactly $r$ from one of the diametrical points. See Figure~\ref{fig:OPT2r} for an example. Since both diametrical points are in $P$, the optimal solution therefore has length at least $D(i) - 2r$. Thus, we get
    \begin{equation*}
        |q'_1(i)q'_2(i)| = D(i) \leq \OPT(i) + 2r
    \end{equation*}
    For $t \in (i, i+1]$, since points in $P$ move with at most unit speed, observe that $|q_1(t)q_2(t)| \leq |q'_1(i)q'_2(i)| + 2$. In addition, $\OPT$ can also change at most $2$ per time unit. As such, we get
    \begin{equation*}
        |q_1(t)q_2(t)| \leq |q'_1(i)q'_2(i)| + 2 \leq \OPT(i) + 2r + 2 \leq \OPT(t) + 2r + 4.\qedhere
    \end{equation*}
\end{proof}

\begin{restatable}{lemma}{hausdorffBound}
    \label{lem:hausdorff-bound}
    While $A(t) = q_1(t)q_2(t)$, the distance from $p$ to $q_1q_2$ is at most $2r\sqrt{2}+4$ for all $p \in P$. 
\end{restatable}
\begin{proof}
    To bound the maximum distance between any point~$p$ and the line piece $q_1q_2$ at any point in time, recall that algorithm~$A$ interpolates to the canonical solution of the time step before. So consider the canonical solutions at times~$i$ and~$i+1$ and corresponding bounding boxes in the orientation of~$q_1'(i)q_2'(i)$ and~$q_1'(i+1)q_2'(i+1)$, respectively (see Figure~\ref{fig:kinetic-speed}c).
    
    At time~$i+1$ the canonical solution of time~$i$ is reached, and the points can have moved at most one unit outside the bounding box of time~$i$, and at time~$i+2$ the points can have moved one unit further outside. Now consider any $t\in [i+1, i+2]$, in which $A$ interpolated from the canonical solution at time~$i$ to the canonical solution at time~$i+1$. For any such~$t$ a point $p$ is at distance at most~$E+2\leq 2r\sqrt{2}+4$ from the canonical solution at time~$i$, which is equal to~$A(i+1)=q_1(i+1)q_2(i+1)$.
    
    Finally, at time~$i+1$ all points must be inside the bounding box in the orientation of the canonical solution at time~$i+1$. At time~$i+2$, the points hence lie at most~$E+1\leq E+2\leq 2r\sqrt{2}+3$ away from the canonical solution at time~$i+1$, which is~$A(i+2) = q_1(i+2)q_2(i+2)$. By linearly interpolating between~$q_1(i+1)q_2(i+1)$ and~$q_1(i+2)q_2(i+2)$, we ensure that no point is more than $2r\sqrt{2}+4$ removed from $q_1q_2$ at any~$t\in [i+1, i+2]$.
    \end{proof}

We can now combine Lemmata~\ref{lem:exist-lowerbound}--\ref{lem:hausdorff-bound} to get the following theorem.

\begin{restatable}{theorem}{stabilityFull}
    \label{thm:stability-theorem}
    Given a set $P$ of points moving with at most unit speed, algorithm $A$ yields a stable approximating segment with $l = 2r + 4$ and $h = 2\sqrt{2} + 4$, for which speed of the endpoints is bounded by $(2r + 1)\sqrt{2} + 2$.
\end{restatable}

\section{Conclusion}
In this paper, we presented an $O(n \log h + h \log^3 h)$ time algorithm to find the shortest representative segment of a point set, improving the previous $O(n \log h + h^{1+\varepsilon})$ time solution. Additionally, we showed how to maintain an approximation of the shortest representative segment in a stable manner, such that its endpoints move with a speed bounded by a linear function in $r$. 

There may be possibilities for improving the running time of our static solution to $O(n \log h + h \log^2 h)$, or even $O(n \log h)$. The $O(h \log^3 h)$ term comes from having to handle $O(h \log h)$ type-5 events in $O(\log^2 h)$ time each. However, it may be possible to show that there are at most $O(h)$ type-5 events, since the conjugate pairs used to bound the number of internal events each have a unique starting point. Additionally, it may be possible to improve the query time of the data structure described in Section~\ref{sec:ds} to $O(\log h)$ time using ideas like fractional cascading, but there is no straightforward way to make this work for the circular query.

% In Appendix~\ref{app:approx} we show that we can compute $(1+\eps)$-approximation in $O(n\log h + h/\eps)$ time by sampling orientations and applying the fixed orientation algorithm.

%%
%% Bibliography
%%

%% Please use bibtex, 

\bibliography{bibliography}
\end{document}